\documentclass[12pt]{article}

\usepackage[top=3.5cm,bottom=3.5cm]{geometry}
\usepackage{amsthm}
\usepackage{amsmath}
\usepackage{amssymb}
\usepackage{amsbsy}
\usepackage{trfsigns}
\usepackage{subfigure}

\usepackage{graphicx}
\usepackage{amsmath,amssymb,amsbsy}
\usepackage{xcolor}

\def\init{0}
\def\EGM{\mathrm{EGM}}
\def\DSP{\mathrm{DSP}}
\def\LSP{\mathrm{LSP}}

\newtheorem{theorem}{Theorem}
\newtheorem{lemma}{Lemma}
\newtheorem{remark}{Remark}
\newtheorem{corollary}{Corollary}

\usepackage{trfsigns}
\title{A computational model for proliferation dynamics of division- and label-structured populations}
\author{J.~Hasenauer, D.~Schittler, and F. Allg{\"o}wer \\ \\
  Institute of Systems Theory and Automatic Control\\
  University of Stuttgart, Germany\\
  www.ist.uni-stuttgart.de \\ \\
}

\begin{document}

\maketitle

\begin{abstract}

In most biological studies and processes, cell proliferation and population dynamics play an essential role. Due to this ubiquity, a multitude of mathematical models has been developed to describe these processes. While the simplest models only consider the size of the overall populations, others take division numbers and labeling of the cells into account. In this work, we present a modeling and computational framework for proliferating cell population undergoing symmetric cell division. In contrast to existing models, the proposed model incorporates both, the discrete age structure and continuous label dynamics. Thus, it allows for the consideration of division number dependent parameters as well as the direct comparison of the model prediction with labeling experiments, e.g., performed with Carboxyfluorescein succinimidyl ester (CFSE). We prove that under mild assumptions the resulting system of coupled partial differential equations (PDEs) can be decomposed into a system of ordinary differential equations (ODEs) and a set of decoupled PDEs, which reduces the computational effort drastically. Furthermore, the PDEs are solved analytically and the ODE system is truncated, which allows for the prediction of the label distribution of complex systems using a low-dimensional system of ODEs. In addition to modeling of labeling dynamics, we link the label-induced fluorescence to the measure fluorescence which includes autofluorescence. For the resulting numerically challenging convolution integral, we provide an analytical approximation. This is illustrated by modeling and simulating a proliferating population with division number dependent proliferation rate.   

\end{abstract}

\textbf{Keyword:}
Proliferating population, label-structured population, flow cytometry, CFSE, division-structured population

\section{Introduction}
\label{sec: Introduction}

Cell proliferation is a central aspect of most biological processes, among others bacterial growth \cite{ZwieteringJon1990,Gyllenberg1986}, immune response \cite{HodgkinLee1996,DeBoerGan2006,LuzyaninaMru2007}, stem cell induced tissue remodeling \cite{GlaucheMoo2009,BuskeGal2011}, and cancer progression \cite{EissingKue2011}. Depending on the biological process, cellular proliferation has different characteristics. Cell division can be symmetric or asymmetric, and the daughter cells may or may not inherit the age of the mother cells (see Figure~\ref{fig: illustration of symmetric and asymmetric division}). While many micro-organisms, such as the budding yeast, grow a daughter cell \cite{ShcheprovaBal2008} which does not inherit the age of the mother cell and is born young, in most multicellular organism the mother cell divides symmetrically into two daughter cells which inherit the age of the mother cell \cite{Hayflick1965}. The latter proliferation type -- which will be the focus of this work -- results in an accumulation of DNA damage and telomere shortening, which may be interpreted as aging of the individual cell. This results in a reduced proliferation potential, a reduced proliferation speed and finally in cell cycle arrest \cite{Hayflick1965,Hayflick1979,GlaucheThi2011}, known as senescence \cite{GerwitzHol2007}. This has been discovered by Hayflick \cite{Hayflick1965} in the 1960s and the upper limit for the number of cell divisions a normal cell can undergo has been termed Hayflick limit.

\begin{figure}[t!]
\centering
\subfigure[Mother cells undergoing symmetric cell division split up into almost identical daughter cells. The difference of daughter cells is merely caused by the stochasticity of DNA replication, resulting in genetic and epigenetic differences, and the stochasticity of cell partitioning, yielding different protein abundances.]{
\includegraphics[scale=1.3]{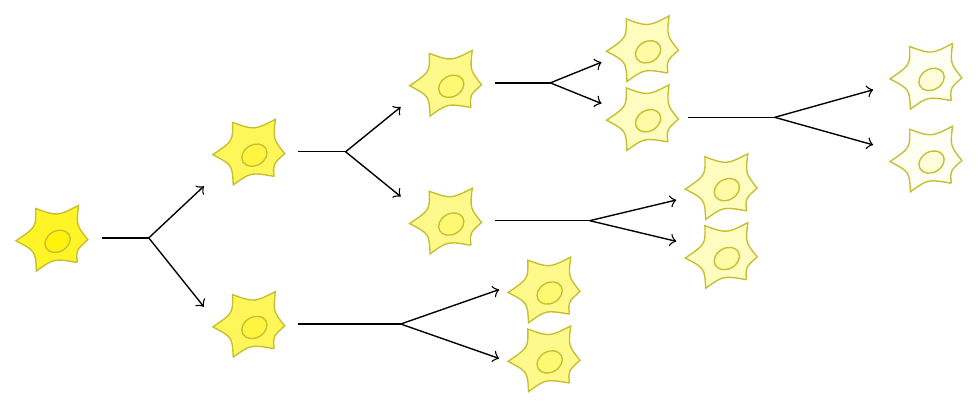}
}
\subfigure[Mother cells undergoing asymmetric cell division yield daughter cells with different cell fates. In the most extreme case, the mother cells grows a daughter cell -- this is called budding. The daughter cell is genetically identical to the mother cell but is born young, meaning that it does not inherit the age of the mother cell.]{
\includegraphics[scale=1.3]{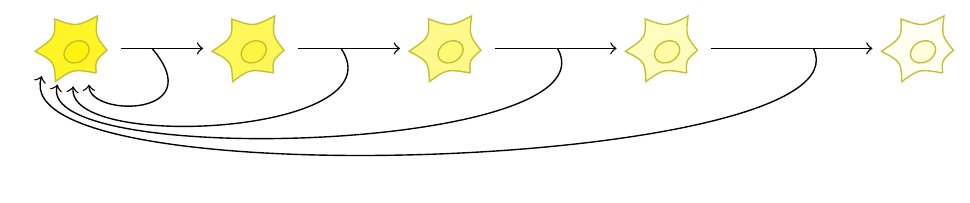}
}
\caption{Illustration of symmetric and asymmetric cell division. The color of the cells indicates the cell's age. In case of symmetric cell division the age strongly correlates with the number of divisions a cell has undergone up to this point.}
\label{fig: illustration of symmetric and asymmetric division}
\end{figure}

A variety of approaches are employed to investigate proliferation, ranging from the analysis of the cell cycle~\cite{SmithMar1973} to the model-based study of population heterogeneity and subpopulations~\cite{GlaucheMoo2009}. Nowadays, for human cell lines especially label-based proliferation assays are used to analyze the proliferation dynamics of cell populations. Common labels are Bromodeoxyuridine (BrdU)~\cite{Gratzner1982} and Carboxyfluorescein succinimidyl ester (CFSE)~\cite{LyonsPar1994}, while mainly the latter is used in recent studies.

CFSE is a fluorescent cell staining dye which stays in cells for a long time and is distributed at cell division approximately equally among daughter cells. Thus, the proliferation of labeled cells results in a  progressive dilution of the dye~\cite{LuzyaninaRoo2009}, as depicted in Figure~\ref{fig: illustration of label dynamics}, and quantitative information about the proliferation dynamics can be gathered using flow cytometry~\cite{HawkinsHom2007}. To determine the proliferation properties of cells, e.g., the rates of cell division and of cell death, from these data, analysis tools are required. The first proposed approaches employ peak detection and devolution~\cite{LuzyaninaMru2007,HawkinsHom2007,NordonNak1999}. Unfortunately, these methods are only applicable if the modes, corresponding to cells with a common division number, are well separated and if the data are not strongly noise corrupted. To overcome these limitations, different model-based approaches have been introduced.

In the literature, mainly three different classes of population models are described: exponential growth models, division-structured population models and label-structured population models. The exponential growth models (EGM) are the simplest ones, and merely describe the number of individuals in a cell population. For this task a one-dimensional ODE, like the Gompertz equation~\cite{ZwieteringJon1990}, is sufficient. While exponential growth models allow the description of the proliferation of many bacterial populations, they are in general not capable of describing the dynamics of human tissue cells. One reason for this is that the cell division and cell death rates are found for many cell systems~\cite{Hayflick1965}, e.g., B cells~\cite{DeBoerGan2006}, T cells~\cite{DeBoerGan2006}, osteoblasts~\cite{KassemAnk1997}, to depend on the division number. To capture these effects, a multitude of division-structured population models (DSP) has been introduced~\cite{DeBoerGan2006,NordonNak1999,RevySos2001,DeenickGet2003,LeePer2008,LeonFar2004,YatesCha2007,Marciniak-CzochraSti2009,StiehlMar2011}. The state variables of these models describe the sizes of the subpopulations, which are defined by a common division number. Hence, these models allow for the consideration of division number dependent properties. Still, these models do not provide information about the label concentrations and thus cannot be compared to data directly but require complicated and error-prone data processing. 

To avoid this, label-structured population models (LSP) are employed~\cite{LuzyaninaRoo2007}. These models describe the evolution of the population density on the basis of a one-dimensional hyperbolic PDE. Hence, they provide predictions for the label distributions at the individual time points and may be fitted to data directly~\cite{LuzyaninaRoo2009,LuzyaninaRoo2007,BanksSut2010,BanksBoc2011}. This renders complex data processing redundant and simplifies the model-data comparison. Still, these models do not allow for a direct consideration of division number dependent parameters. To partly circumvent this problem, complex dependencies of the cell division and cell death rate on time and label concentration are introduced~\cite{BanksSut2010}. These are neither intuitive nor easy to interpret. Furthermore, the simulation of label-structured population models is computationally demanding and requires discretization, entailing further problems.

In the following a model is presented and analyzed which combines the division-structured population models and the label-structured population models and thereby overcomes their individual shortcomings. This population model, which we termed division- and label-structured population model (DLSP), is based on our own work~\cite{SchittlerHas2011a}. The same model has later also been used in \cite{BanksThom2012,Thompson2012} for parameter estimation, including slight modification. Here we provide the first rigorous in-depth assessment its properties.

The DLSP model is introduced in Section~\ref{sec: Modeling DLSP} and incorporates both aspects: Discrete changes of the cell division number due to cell divisions and continuous dynamics of the label distribution. The overall model is a system of coupled partial differential equations. We discuss how this system of PDEs can be split up into two decoupled parts in Section~\ref{sec: Analysis of DLSP model}, namely a single PDE and a set of ODEs, which significantly simplifies the solution. The obtained model is reduced further by truncation of the state space. This truncation and the resulting truncation error can be controlled using the a priori error bound which we derive. As the proposed model unifies the existing models, we outline the relations of the models in Section~\ref{sec: Comparison of different proliferation models}. In Section~\ref{sec: Example}, the method is employed to study a population model with division number dependent division rates and an analysis of the computational complexity of the model is performed. The paper is concluded in Section~\ref{sec: conclusion}.

\begin{figure}[t!]
\centering
\includegraphics[scale=1]{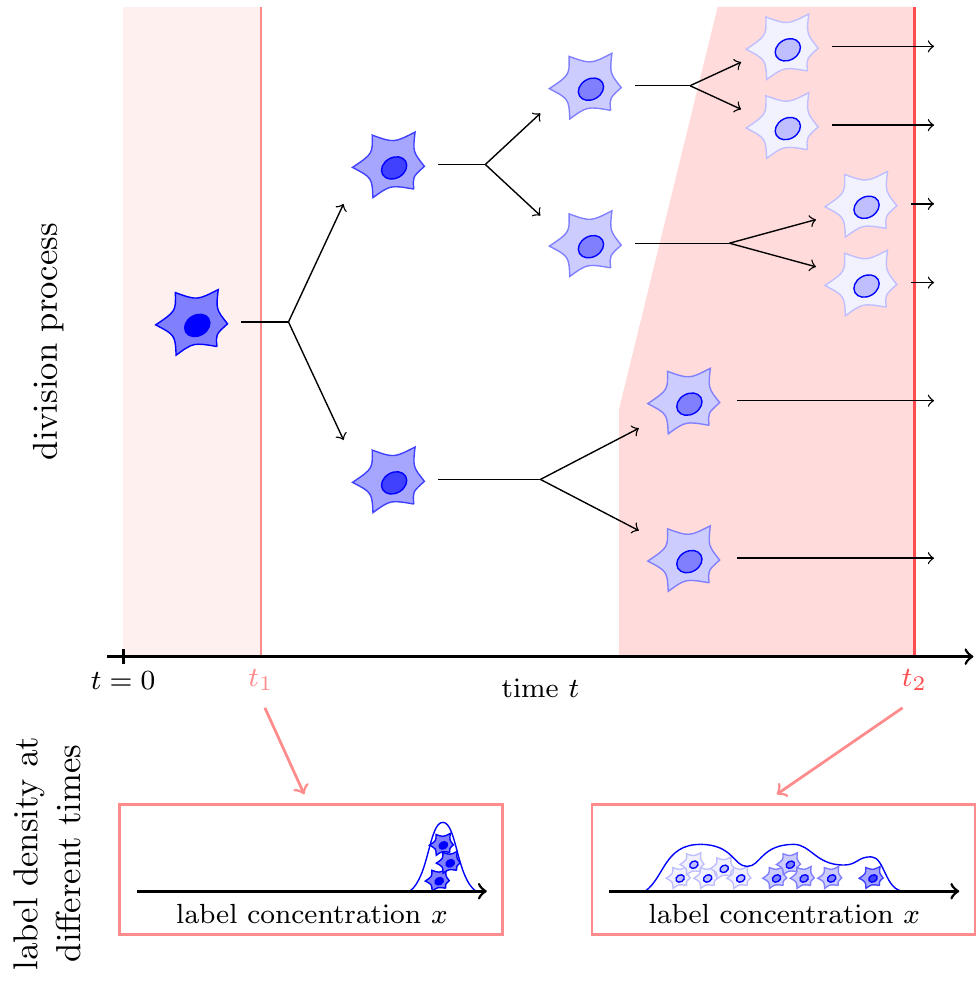}
\caption{Illustration of label dilution due to cell division. The division process results in halving of the concentration at each cell division (top), and the label intensity distribution within the cell populations (bottom). The latter one is accessible, e.g., via labeling with CFSE.}
\label{fig: illustration of label dynamics}
\end{figure}


\section{Modeling division- and label-structured populations}
\label{sec: Modeling DLSP}

As outlined above, the study of proliferation dynamics in cell populations using labeling methods requires the consideration of two important distinct features:
\begin{itemize}
  \item[\tiny{$\blacksquare$}] the label concentration $x$ and
  \item[\tiny{$\blacksquare$}] the number of cell divisions $i$ a cell has undergone.
\end{itemize}
The importance of the label concentration $x \in \mathbb{R}_+$ (with $\mathbb{R}_+ := [0,\infty)$) arises from the fact that this is the quantity which can be observed, e.g., using flow cytometry or microscopy~\cite{HawkinsHom2007}. On the other hand, a direct observation of the number of cell divisions $i \in \mathbb{N}_0$ a cell has undergone is in general not possible, though the division number often plays a crucial role within the model. A cell which has divided once is expected to have different properties, e.g. a different division rate, than a cell which has already divided several dozen times~\cite{Hayflick1965,KassemAnk1997}.

In this paper we propose a model which captures both features of cells, distinct division numbers as well as distinct label concentrations among cells. Therefore, instead of a single PDE model describing the label dynamics of the overall population, a PDE model is defined for every subpopulation. Thereby, the $i$th subpopulation contains the cells which have divided $i$ times. Cell division generates a flux from subpopulation $i$ to subpopulation $i+1$, thus inducing coupling. The system of coupled PDEs is given by
\begin{equation}
\begin{aligned}
	i = 0: \quad &
	\frac{\partial N_0(t,x)}{\partial t} + \frac{\partial (\nu(t,x) N_0(t,x))}{\partial x} =
	- \left( \alpha_0(t) + \beta_0(t) \right) N_0(t,x)\\[2ex]
	\forall i \geq 1: \quad &
	\frac{\partial N_i(t,x)}{\partial t} + \frac{\partial (\nu(t,x) N_i(t,x))}{\partial x} =
	- \left( \alpha_i(t) + \beta_i(t) \right) N_i(t,x) \\
	& \hspace{6.5cm}
	+ 2 \gamma \alpha_{i-1}(t) N_{i-1}(t,\gamma x),
\label{eq: coupled PDE model for cell population}
\end{aligned}
\end{equation}
with initial conditions
\begin{align*}
i = 0: N_0(0,x) \equiv N_{0,\init}(x), \quad \forall i \geq 1: N_i(0,x) \equiv 0.
\end{align*}
In this system, $N_i(t,x): \mathbb{R_+} \times \mathbb{R_+} \rightarrow \mathbb{R_+}$ denotes the label density in the $i$th subpopulation at time $t$. The structure of the models for the individual subpopulations is highly similar to a single PDE which is employed in label-structured models~\cite{LuzyaninaRoo2007}. The fluxes influencing the label distribution $N_i(t,x)$ are: 
\begin{itemize}
  \item[\tiny{$\blacksquare$}] 
$\partial (\nu(t,x) N_i(t,x))/\partial x$, decay of label $x$ in each cell with label loss rate $\nu(t,x)$.
  \item[\tiny{$\blacksquare$}] 
$- \left( \alpha_i(t) + \beta_i(t) \right) N_i(t,x)$, disappearance of cells from the $i$th subpopulation due to cell division with rate $\alpha_i(t)$ and due to cell death with rate $\beta_i(t)$.
  \item[\tiny{$\blacksquare$}] 
$2 \gamma \alpha_{i-1}(t) N_{i-1}(t,\gamma x)$, appearance of two cells due to cell division in the $(i-1)$th subpopulation with division rate $\alpha_{i-1}(t)$. The factor $\gamma \in (1,2]$ is the rate of label dilution due to cell division (cf. \cite{BanksSut2010,LuzyaninaRoo2009}).
\end{itemize}
It has to be emphasized that the division rates $\alpha_i(t): \mathbb{R}_+ \rightarrow \mathbb{R}_+$ as well as the death rates ${\beta_i(t): \mathbb{R}_+ \rightarrow \mathbb{R}_+}$ may depend on division number $i$ and time $t$. To ensure existence and uniqueness of the solutions we require $\alpha_i(t),\beta_i(t) \in \mathcal{C}^1$. As it is assumed that the labeling does not affect cell function, we do not allow $\alpha_i$ and $\beta_i$ to depend on the label concentration $x$. Furthermore, only label loss rates are considered which follow a linear degradation
\begin{align}
	\nu(x) = -k(t) x.
	\label{eq: label loss rate}
\end{align}
The time dependence of the degradation rate may be arbitrary, but mainly constant degradation processes \cite{LuzyaninaRoo2007,LuzyaninaRoo2009,BanksSut2010}, $k(t) = k \; (\text{const.})$, or Gompertz decay processes, $k(t) = c_1 e^{-c_2 t}$\cite{BanksBoc2011}, are used.

Note that by construction model~\eqref{eq: coupled PDE model for cell population} provides information about cell numbers and label density for the overall as well as for individual subpopulations. Hence, it combines advantages of common ODE models~\cite{DeBoerGan2006,RevySos2001,DeenickGet2003} and common PDE models~\cite{LuzyaninaRoo2009,LuzyaninaRoo2007,BanksSut2010} of cell populations and permits for more biologically plausible degrees of freedom than both of them. In detail, the available information are:

\textit{Number of cells in the subpopulations:}
Given $N_i(t,x)$, the number of cells contained in the $i$th subpopulation can be computed as
\begin{align}
	\bar{N}_i(t) = \int_{\mathbb{R}_+} N_i(t,x) dx.
	\label{eq: def - number of cells}
\end{align}
This number of cells may help to understand the relative contribution of subpopulations to the overall population.

\textit{Normalized label density in the subpopulations:}
Given $N_i(t,x)$, the label density within the $i$th subpopulation can be computed as
\begin{align}
	n_i(t,x) = 
	\left\{
	\begin{array}{cl}
	\dfrac{N_i(t,x)}{\bar{N}_i(t)} & \text{for } \bar{N}_i(t) > 0 \\[2ex]
	0 & \text{otherwise}.
	\end{array}
	\right.
	\label{eq: def - label density}
\end{align}
The normalized label density provides the probability of finding a cell within the $i$th subpopulation with label concentration ${\xi \in [x,x+\Delta x]}$,
\begin{align}
	\mathrm{Prob}(\xi \in [x,x+\Delta x]) = \int_x^{x+\Delta x} n_i(t,x) dx.
\end{align}

Besides the properties of the subpopulations, the model permits also the analysis of the properties of the overall population. The unnormalized label density in the overall cell population $M(t,x): \mathbb{R}_+ \times \mathbb{R}_+ \rightarrow \mathbb{R}_+$ is given by
\begin{align}
	M(t,x) = \sum_{i \in \mathbb{N}_0} N_i(t,x).
	\label{eq: label density in overall population}
\end{align}
From $M(t,x)$ the overall population size
\begin{align}
	\bar{M}(t) = \int_0^\infty M(t,x) dx = \sum_{i \in \mathbb{N}_0} \bar{N_i}(t)
	\label{eq: size of overall population}
\end{align}
and the normalized label density in the overall population
\begin{align}
	m(t,x) = 	\left\{
	\begin{array}{cl}
	\dfrac{M(t,x)}{\bar{M}(t)} = \dfrac{\sum_{i \in \mathbb{N}_0} N_i(t,x)}{\sum_{i \in \mathbb{N}_0} \bar{N_i}(t)}& \text{for } \bar{M}(t) > 0 \\[2ex]
	0 & \text{otherwise}
	\end{array}
	\right.
	\label{eq: label density of overall population}
\end{align}
can be derived. These are the two experimentally observable variables of the system. The overall population size $\bar{M}(t)$ can be determined by cell counting, while the population density $m(t,x)$ can be assessed by the labeling with CFSE or BrdU. By combining these two, $M(t,x)$ can be reconstructed. As there is currently no direct cell division marker available, experimental assessment of the subpopulation sizes or of the label distribution within the subpopulations is in general not feasible. All common experimental techniques only provide the marginalization over the division number~$i$~\cite{LuzyaninaRoo2009,LuzyaninaRoo2007,BanksSut2010}.


\section{Analysis of division- and label-structured population model}
\label{sec: Analysis of DLSP model}

Besides the advantages the DLSP model offers, its potential drawback is its complexity. The model is a system of coupled PDEs, which are in general difficult to analyze, and their simulation is often computationally demanding or even intractable. In the following it is shown that these problems can be solved for the DLSP model \eqref{eq: coupled PDE model for cell population}. The approach presented allows to efficiently compute the solution of the DLSP model, without solving a system of coupled PDEs.

\subsection{Solution of the DLSP via decomposition}
\label{subsec: Solution of the DLSP via decomposition}

In order to provide an efficient method for computing the solution of~\eqref{eq: coupled PDE model for cell population}, we define the initial number of cells
\begin{align}
	\bar{N}_{0,\init} = \int_{\mathbb{R}_+} N_{0,\init}(x) dx
\end{align}
and the initial label density
\begin{align}
	n_{0,\init}(x) = 
	\left\{
	\begin{array}{cl}
	\dfrac{N_{0,\init}(x)}{\bar{N}_{0,\init}} & \text{for} \; \bar{N}_{0,\init} > 0 \\[2ex]
	0 & \text{otherwise},
	\end{array}
	\right.
\end{align}
according to~\eqref{eq: def - number of cells} and~\eqref{eq: def - label density}. Given these definitions the following theorem holds:

\begin{theorem}
\label{theorem: decomposition}
The solution of model \eqref{eq: coupled PDE model for cell population} is
\begin{align}
	\forall i: \quad N_i(t,x) = \bar{N}_i(t) n_i(t,x)
	\label{eq: ansatz for solution}
\end{align}
in which: \\[1ex]
(i) \hspace{2mm} $\bar{N}_i(t)$ is the solution of the system of the ODE:
\begin{align}
\begin{split}
	i = 0: \hspace{2mm} &
	\frac{d \bar{N}_0}{d t} = - \left( \alpha_0(t) + \beta_0(t) \right) \bar{N}_0,\\
	\forall i \geq 1: \hspace{2mm} &
	\frac{d \bar{N}_i}{d t} = - \left( \alpha_i(t) + \beta_i(t) \right) \bar{N}_i + 2 \alpha_{i-1}(t) \bar{N}_{i-1}
	\label{eq: ODE part of ansatz} 
\end{split}
\end{align}
with initial conditions: $\bar{N}_0(0) = \bar{N}_{0,\init}$ and $\forall i \geq 1: \bar{N}_i(0) = 0$. \\[1ex]
(ii) \hspace{2mm} $n_i(t,x)$ is the solution of the PDE:
\begin{align}
\forall i: \hspace{2mm} &\frac{\partial n_i(t,x)}{\partial t} - k(t) \frac{\partial (x n_i(t,x))}{\partial x} = 0
\label{eq: PDE part of ansatz}
\end{align}
with initial conditions $\forall i:$ $n_i(0,x) \equiv \gamma^i n_{0,\init}(\gamma^i x)$.
\end{theorem}
The state variables $\bar{N}_i(t)$ and $n_i(t,x)$ of the ODE system and the PDEs correspond to the number of cells~\eqref{eq: def - number of cells} and the label density~\eqref{eq: def - label density} in the $i$th subpopulation, respectively.

\begin{proof}
To prove that Theorem~\ref{theorem: decomposition} holds, \eqref{eq: ansatz for solution} - \eqref{eq: PDE part of ansatz} are inserted in \eqref{eq: coupled PDE model for cell population} and it is shown that the resulting equation holds.  The proof is only shown for $i\geq1$, since the case $i=0$ can be treated analogously. Furthermore, for notational simplicity the dependence of $n_i(t,x)$, $\alpha_i(t)$, and $\beta_i(t)$ on $t$ and $x$ is omitted where not required.

Inserting \eqref{eq: ansatz for solution} in \eqref{eq: coupled PDE model for cell population} for $i \geq 1$ yields
\begin{align}
	\frac{\partial \bar{N}_i n_i}{\partial t} - k \frac{\partial (x \bar{N}_i n_i)}{\partial x} =
	- \left( \alpha_i + \beta_i \right) \bar{N}_i n_i(t,x) + 2\gamma \alpha_{i-1} \bar{N}_{i-1} n_{i-1}(t,\gamma x).
\label{eq: proof - step 1}
\end{align}
The left hand side of this equation can be reformulated:
\begin{equation}
\begin{aligned}
	\frac{\partial \bar{N}_i n_i}{\partial t} - k \frac{\partial (x \bar{N}_i n_i)}{\partial x}
	&= \frac{d\bar{N}_i}{dt} n_i + \bar{N}_i \frac{\partial n_i}{\partial t} - k \bar{N}_i \frac{\partial (x n_i)}{\partial x} \\
	&= \frac{d\bar{N}_i}{dt} n_i + \bar{N}_i \left(\frac{\partial n_i}{\partial t} - k \frac{\partial (x n_i)}{\partial x}\right)
	\overset{\eqref{eq: PDE part of ansatz}}{=} \frac{d\bar{N}_i}{dt} n_i.
\end{aligned}
\end{equation}
By inserting this result in \eqref{eq: proof - step 1} and substituting $d\bar{N}_i/dt$ with \eqref{eq: ODE part of ansatz}, we obtain
\begin{equation}
\begin{aligned}
	&\left(- \left( \alpha_i + \beta_i \right) \bar{N}_i + 2 \alpha_{i-1} \bar{N}_{i-1}\right) n_i(t,x) = \\
	& \hspace{3cm} - \left( \alpha_i + \beta_i \right) \bar{N}_i n_i(t,x) + 2\gamma \alpha_{i-1} \bar{N}_{i-1} n_{i-1}(t,\gamma x),
\end{aligned}
\end{equation}
which can be simplified to
\begin{align}
	n_i(t,x) = \gamma n_{i-1}(t,\gamma x).
	\label{eq: proof - final step}
\end{align}
It can be proven that this last equality holds, e.g., by using the analytical solution of \eqref{eq: PDE part of ansatz}, which can be found below. This yields that \eqref{eq: proof - final step} holds which concludes the proof of Theorem~\ref{theorem: decomposition}.
\end{proof}

\begin{remark}
Note that it can be verified that \eqref{eq: proof - final step} holds if and only if the label loss rate $\nu(t,x)$ is linear in $x$. \end{remark}

With Theorem~\ref{theorem: decomposition}, the original system of coupled PDEs can be decomposed into a system of ODEs \eqref{eq: ODE part of ansatz} and a set of decoupled PDEs \eqref{eq: PDE part of ansatz}. This means that the size of the individual subpopulations can be decoupled from the label dynamics. This already tremendously simplifies the analysis, but a further simplification is possible:

\begin{corollary}
\label{corollary: decomposition with exact solution for PDE}
The solution of model \eqref{eq: coupled PDE model for cell population} is
\begin{align}
	\forall i: \quad N_i(t,x) = \bar{N}_i(t) \gamma^i e^{-\int_0^t k(\tau) d\tau} n_{0,\init}(\gamma^i e^{\int_0^t k(\tau) d\tau} x),
	\label{eq: ansatz for solution including PDE solution}
\end{align}
in which $\bar{N}_i(t)$ is the solution of the ODE~\eqref{eq: ODE part of ansatz}.
\end{corollary}

\begin{proof}
To prove Corollary~\ref{corollary: decomposition with exact solution for PDE} note that the PDE \eqref{eq: PDE part of ansatz} is linear. Thus, the method of characteristics~\cite{Evans1998} can be employed to obtain an analytical solution (Appendix~\ref{app sec: Analytical solution of PDE}). This yields
\begin{align}
\forall i: n_i(t,x) = \gamma^i e^{-\int_0^t k(\tau) d\tau} n_{0,\init}(\gamma^i e^{\int_0^t k(\tau) d\tau} x),
\end{align}
which can be inserted into~\eqref{eq: ansatz for solution}, proving Corollary~\ref{corollary: decomposition with exact solution for PDE}. 
\end{proof}

The general solution $n_i(t,x)$ simplifies in cases of specific choices for $k(t)$. A constant degradation rate yields 
\begin{align}
\forall i: n_i(t,x) = \gamma^i e^{-k t} n_{0,\init}(\gamma^i e^{k t} x),
\end{align}
while for a Gompertz decay process one obtains,
\begin{align}
\forall i: n_i(t,x) = \gamma^i e^{- \frac{c_1}{c_2} (1 - e^{-c_2 t})} n_{0,\init}(\gamma^i e^{\frac{c_1}{c_2} (1 - e^{-c_2 t})} x).
\end{align}
Corollary~\ref{corollary: decomposition with exact solution for PDE} provides a solution for any label degradation rates, including those considered in~\cite{BanksThom2012,Thompson2012}.

By solving the decoupled PDEs analytically, the solution of the DLSP model can be obtained in terms of the solution of a system of ODEs. This reduces the complexity drastically and enables also a compact representation of the overall label density $M(t,x)$:

\begin{corollary}
\label{corollary: solution of overall label density}
The overall label density~\eqref{eq: label density in overall population} is
\begin{equation}
\begin{aligned}
	M(t,x)
	&= \sum_{i \in \mathbb{N}_0} \bar{N}_i(t) n_i(t,x) = \sum_{i \in \mathbb{N}_0} \bar{N}_i(t) \gamma^i e^{-\int_0^t k(\tau) d\tau} n_{0,\init}(\gamma^i e^{\int_0^t k(\tau) d\tau} x),
	\label{eq: solution of overall label density}
\end{aligned}
\end{equation}
in which $\bar{N}_i(t)$ is the solution of the ODE~\eqref{eq: ODE part of ansatz}.
\end{corollary}

\begin{proof}
By substituting~\eqref{eq: ansatz for solution including PDE solution} into~\eqref{eq: label density in overall population}, Corollary~\ref{corollary: solution of overall label density} is proven.
\end{proof}

Given Corollary~\ref{corollary: decomposition with exact solution for PDE} and~\ref{corollary: solution of overall label density}, it is apparent that merely the ODE system~\eqref{eq: ODE part of ansatz} has to be solved in order to compute the solution of the DLSP. This problem is approached in the remainder of this section.

\subsection{Calculation of the subpopulation sizes}
\label{subsec: Calculation of the subpopulation sizes}

In order to solve ODE system~\eqref{eq: ODE part of ansatz}, we note that the change of subpopulation $i$ only depends on the size of subpopulation $i-1$. This chain-like structure enables the solution of $\bar{N}_{i}(t)$ via recursion. By doing so, analytical solutions for the ODE system have been found for two cases~\cite{LuzyaninaMru2007,RevySos2001}:

\begin{lemma}
\label{lem: ODE solution case 1}
Given that $\forall i \in \mathbb{N}_0: \alpha_i(t) = \alpha \geq 0\ \wedge \ \beta_i(t) = \beta > 0$, the solution of \eqref{eq: ODE part of ansatz} is:
\begin{align}
	\bar{N}_i(t) = \frac{(2 \alpha t)^i}{i!} e^{-(\alpha + \beta)t} \bar{N}_{0,\init}.
	\label{eq:solution ODE case 1}
\end{align}
\end{lemma}
This result has been derived in \cite{RevySos2001}, where the authors studied this ODE system to model the number of cells that have undergone a certain number of divisions, without modeling label dynamics. The derivation as provided in Appendix~\ref{app sec: Analytical solution of ODE - Case 1} is generalized for later use.

\begin{lemma}
\label{lem: ODE solution case 2}
Given that $\forall i \in \mathbb{N}_0: \alpha_i(t) = \alpha_i \geq 0 \ \wedge \ \beta_i(t) = \beta_i > 0$ and $\forall i,j \in \mathbb{N}_0, i \neq j: \alpha_i + \beta_i \neq \alpha_j + \beta_j$, then the solution of \eqref{eq: ODE part of ansatz} is:
\begin{align}
\begin{split}
	i = 0: \hspace{1mm} &
	\bar{N}_0(t) = e^{-(\alpha_0 + \beta_0)t} \bar{N}_{0,\init} \\
	\forall i \geq 1: \hspace{1mm} &
	\bar{N}_i(t) = 2^i \left(\prod_{j=1}^{i} \alpha_{j-1}\right) D_i(t) \bar{N}_{0,\init}
\end{split}
\label{eq:solution ODE case 2}
\end{align}
in which
\begin{align*}
	D_i(t) = \sum_{j=0}^{i}
	\left[
	\left(\prod_{\substack{k=0\\k\neq j}}^i ((\alpha_k + \beta_k) - (\alpha_j + \beta_j))\right)^{-1} e^{-(\alpha_j + \beta_j)t}
	\right].
\end{align*}
\end{lemma}
Solution~\eqref{eq:solution ODE case 2} was first stated in \cite{LuzyaninaMru2007} and for completeness the proof is provided in Appendix~\ref{app sec: Analytical solution of ODE - Case 2}. It basically employs mathematical induction in the frequency domain, exploiting properties of the partial fraction under the provided assumptions. Despite the prerequisites, this result is quite powerful as for almost all cases of time invariant division number dependent parameters $\alpha_i$ and $\beta_i$ the ODE system~\eqref{eq: ODE part of ansatz} can be solved analytically.

In cases in which neither prerequisites for Lemma~\ref{lem: ODE solution case 1} nor~\ref{lem: ODE solution case 2} hold, then the solution of~\eqref{eq: ODE part of ansatz} can still be computed using numerical integration. This is possible if only the sizes of the first $S$ subpopulations $\bar{N}_0(t)$, $\bar{N}_1(t)$, $\ldots$, $\bar{N}_{S-1}(t)$, are of interest, where $S$ is finite. 

\subsection{Truncation of division numbers in the population model}
\label{subsec: Truncation of population model}

In Section~\ref{subsec: Solution of the DLSP via decomposition} a decomposition approach has been described to decouple the size of the subpopulations from the label distribution in the individual subpopulations. While this simplifies the computation of the properties of individual subpopulations drastically, the analysis of the overall label density and of the overall population size still requires the calculation of an infinite sum~\eqref{eq: solution of overall label density}. Even in cases for which the individual subpopulation sizes are available analytically (see~\eqref{eq:solution ODE case 1} and~\eqref{eq:solution ODE case 2}), we could not derive a closed form solution for $M(t,x)$. Therefore, in this section we present a method to find an approximation of $M(t,x)$ of the form
\begin{align}
	\hat{M}_S(t,x) &= \sum_{i = 0}^{S-1} N_i(t,x) = \sum_{i = 0}^{S-1} \bar{N}_i(t) n_i(t,x)
	\label{eq: approximated solution of overall label density}
\end{align}
with truncation index $S\geq 0$. Instead of considering an infinite number of subpopulations, only the first $S$ subpopulations are taken into account. While it might be argued that a bound $S$ can be determined from experimental data collected in proliferation assays~\cite{BanksThom2012,Thompson2012}, this is not true for long times. In case of long observation intervals, the autofluorescence -- which will be discussed in Section~\ref{sec: Computation of measured label distribution} -- avoids an estimation of $S$. Thus, reliable selection rules for the truncation index $S$ are necessary.

In order to approximate $M(t,x)$ with arbitrary precision by the truncated sum $\hat{M}_S(t,x)$, convergence of~\eqref{eq: label density in overall population} and~\eqref{eq: size of overall population} with respect to the subpopulation index $i$ is required and can be proven:

\begin{theorem}
\label{theorem: convergence}
	The sums~\eqref{eq: label density in overall population} converge for any finite time $T$, if there exist
	\begin{equation}
	\begin{aligned}
	\alpha_{\sup} &= \sup_{t \in [0,T], i \in \mathbb{N}_0} \alpha_i(t) \geq 0, \\
	\alpha_{\inf} &= \inf_{t \in [0,T], i \in \mathbb{N}_0} \alpha_i(t) \geq 0, \\
	\beta_{\inf} &= \inf_{t \in [0,T], i \in \mathbb{N}_0} \beta_i(t) > 0.
	\end{aligned}
	\end{equation}
\end{theorem}

The proof of Theorem~\ref{theorem: convergence} is provided in Appendix~\ref{app sec: Proof of convergence Theorem}. It employs a system of ODEs of which its states are an upper bound for the states of~\eqref{eq: ODE part of ansatz}, and which can be solved analytically. Given these upper bounds the comparison theorem for series~\cite{Knopp1964} can be used to verify convergence. Note that Theorem~\ref{theorem: convergence} is powerful as it holds for all biological plausible functions $\alpha_i(t)$ and $\beta_i(t)$.

Given convergence the question arises how large the truncation index $S$ must be to ensure a predefined error bound at a given time. For the considered system it can be shown that:

\begin{theorem}
\label{theorem: truncation error}
	Given a truncation index $S$ and a time $T$, as well as  $\alpha_{\inf}$, $\alpha_{\sup}$, and $\beta_{\inf}$ as defined in Theorem~\ref{theorem: convergence}, the truncation error is upper bounded by $E_S(T)$:
 	\begin{equation}
		\frac{||M(T,x) - \hat{M}_S(T,x)||_1}{||M(0,x)||_1} \leq E_S(T) = \left(e^{2 \alpha_{\sup} T} - \sum_{i = 0}^{S-1} \frac{(2 \alpha_{\sup}T)^i}{i!}\right) e^{-(\alpha_{\inf}+\beta_{\inf})T}.
		\label{eq: upper bound for truncation error}
	\end{equation}		
\end{theorem}

To prove Theorem~\ref{theorem: truncation error}, we show that $||M(t,x) - \hat{M}(t,x)||_1 = \sum_{i=S+1}^\infty \bar{N}_i(t)$. This sum can be upper bounded for all biologically plausible functions $\alpha_i(t)$ and $\beta_i(t)$ using the ODE system employed to verify Theorem~\ref{theorem: convergence}. The full proof is provided in Appendix~\ref{app sec: Proof of truncation error Theorem}. Note that, if $\forall i \in \mathbb{N}_0: \alpha_i(t) = \alpha \ \wedge \ \beta_i(t) = \beta$, the bound~\eqref{eq: upper bound for truncation error} is precise and equality holds.

\begin{remark}
	In this work we considered an error bound which is relative to the initial condition. This is reasonable as for this system the superposition principle holds and the relative truncation error $\frac{||M(T,x) - \hat{M}_S(T,x)||_1}{||M(0,x)||_1}$ is thus independent of $||M(0,x)||_1$.
\end{remark}

Given Theorem~\ref{theorem: truncation error}, an upper bound $S$ can be derived which ensure that a relative error is bounded by $\epsilon$:

\begin{corollary}
\label{corollary: bound for S}
	Assuming that $\alpha_{\inf}$, $\alpha_{\sup}$, and $\beta_{\inf}$ exist as defined in Theorem~\ref{theorem: convergence}, the error bound
 	\begin{equation}
		\frac{||M(T,x) - \hat{M}_S(T,x)||_1}{||M(0,x)||_1} \leq \epsilon
		\label{eq: error bound}
	\end{equation}		
	holds if 
 	\begin{equation}
		\left(e^{2 \alpha_{\sup} T} - \sum_{i = 0}^{S-1} \frac{(2 \alpha_{\sup}T)^i}{i!}\right) e^{-(\alpha_{\inf}+\beta_{\inf})T} \leq \epsilon.
		\label{eq: required truncation index}
	\end{equation}		
\end{corollary}

\begin{proof}
Corollary~\ref{corollary: bound for S} follows directly from Theorem~\ref{theorem: truncation error}, using $\frac{||M(T,x) - \hat{M}(T,x)||_1}{||M(0,x)||_1} \leq E_S(T) \leq \epsilon$.
\end{proof}

Despite the generality of Theorem~\ref{theorem: truncation error} and Corollary~\ref{corollary: bound for S} for the considered system class, it suffers the small disadvantage that no explicit expression for $S$ has been found. Rather, the minimum truncation index $S$ which is required to ensure a certain error bound has to be found iteratively by increasing or decreasing $S$ based on the current error. Fortunately, this search is computationally cheap as it is not necessary to solve a system of ODEs or PDEs, but the error bound is available analytically.

A study of the a priori error bound~\eqref{eq: required truncation index} shows that if the acceptable relative error $\epsilon$ is kept constant, the truncation index $S$ grows monotonically as function of the final simulation time. This is due to the exponential growth of $e^{2 \alpha_{\sup} T}$ vs. the polynomial growth $\sum_{i = 0}^{S-1} \frac{(2 \alpha_{\sup}T)^i}{i!}$. Merely for cases in which $2 \alpha_{\sup} \leq \alpha_{\inf}+\beta_{\inf}$, $S$ does not have to increase arbitrarily over time but stays bounded, as under these conditions the population dies out. Note that the increase of $S$ is often not critical. Due to label dilution in general only the first seven or eight cell divisions can be observed~\cite{HawkinsHom2007}, which limits the timespan of interest and therefore the required truncation index~$S$.

Aside from an approximation of the population density $M(t,x)$, also approximations for the overall population size $\bar{M}(t)$ and of the normalized overall label density $m(t,x)$ may be necessary to compare model predictions to measurements. Accordingly to~\eqref{eq: size of overall population} and~\eqref{eq: label density of overall population}, plausible choices for these approximations are
\begin{align}
	\hat{\bar{M}}_S(t) = \int_{\mathbb{R}_+} \hat{M}_S(t,x) dx \quad \text{and} \quad
	\hat{m}_S(t,x) = \frac{\hat{M}_S(t,x)}{\hat{\bar{M}}_S(t)}.
\end{align}
Theorems~\ref{theorem: convergence} and~\ref{theorem: truncation error} can be extended to verify convergence and determine truncation errors for these quantities. For $\hat{\bar{M}}_S(t)$ this is straightforward, while for $\hat{m}_S(t,x)$ it is slightly more complicated. The proofs are not provided here as this is beyond the scope of this work and would reduce the readability.
\\[3ex]
To summarize, in this section the DLSP model has been analyzed in-depth. We have shown that for a very general class of division and death rates $\alpha_i(t)$ and $\beta_i(t)$, the solution of the DLSP can be computed by solving a system of ODEs. This ODE system has an analytical solution for a rather general class of time independent parameterizations. By determining rigorous error bounds, we furthermore enable the calculation of the required truncation index to achieve a predefined precision. As shown later, this will allow for many systems to predict the population response employing a low-dimensional ODE system.

\section{Comparison of different proliferation models}
\label{sec: Comparison of different proliferation models}

In the last section we have analyzed the DLSP model and outlined a method to solve it. The question which remained open is how the DLSP model and its solution relate to existing population models for cell proliferation. To answer this question we confine ourselves to the in our opinion most common models, the exponential growth model (EGM), the division-structured population model (DSP) and the label-structured population model (LSP):
\begin{itemize}
\item[] EGM: An ODE describing the dynamics of the overall population size~\cite{ZwieteringJon1990}.  
\item[] DSP: A system of ODEs describing the dynamics of the number of cells contained in the individual subpopulations, where the subpopulations are defined via a common number of cell divisions~\cite{RevySos2001,DeBoerGan2006}.
\item[] LSP: A PDE describing the dynamics of the label density in the overall population~\cite{BanksSut2010,LuzyaninaRoo2009,LuzyaninaRoo2007}.
\end{itemize}
These models are used in many more publications than cited here and various extensions of these models exist.

\subsection{Relation between EGM and DLSP}
\label{ref: Relation between DLSP and EGM}
The EGM is the simplest available model which describes population dynamics. It has only one state variable, which corresponds to the size of the overall cell population. In general, the EGM is written as
\begin{equation}
	\frac{d\bar{M}^{\EGM}}{dt} = \phi(t) \bar{M}^{\EGM}(t), \quad \bar{M}^{\EGM}(0) = \bar{M}_\init^{\EGM},
\end{equation}
in which $\phi(t)$ is the effective growth rate. A common choice is $\phi(t) = e^{\phi_1 - \phi_2 t}$ which results in a Gompertz equation~\cite{ZwieteringJon1990}.

As the EGM only describes the overall population size, it is contained in the DLSP. By choosing $\alpha_i(t) = \phi(t)$, $\beta_i(t) = 0$ and $\bar{N}_{0,\init} = \bar{M}_\init^{\EGM}$, the overall population size $\bar{M}(t)$ predicted by the DLSP is equivalent to $\bar{M}^{\EGM}(t)$. This can be shown using the time derivative of $\bar{M}$,
\begin{equation}
	\frac{d\bar{M}}{dt} = \sum_{i \in \mathbb{N}_0} \frac{d\bar{N}_i}{dt} = \phi(t) \sum_{i \in \mathbb{N}_0} \bar{N}_i(t) = \phi(t) \bar{M}(t),
\end{equation}
which has the initial condition $\bar{M}_\init = \sum_{i \in \mathbb{N}_0} \bar{N}_{i,\init} = \bar{M}_\init^{\EGM}$.

\subsection{Relation between DSP and DLSP}
\label{ref: Relation between DLSP and DSP}
In contrast to the EGM, the DSP resolves the subpopulations, and the state variables $\bar{N}_i^{\DSP}(t)$ correspond to the number of cells which have divided $i$ times. To our knowledge this model has first been proposed in~\cite{RevySos2001} and its most common form is equal to~\eqref{eq: ODE part of ansatz}. Thus, the DSP is contained in the DLSP and is obtained by marginalization over the label concentration $x$. Actually, according to Theorem~\ref{theorem: decomposition}, a DSP model is solved to compute the solution of the DLSP. As for the PDE component of the DLSP an analytical expression can be derived (Corollary~\ref{corollary: decomposition with exact solution for PDE}), solving the DLSP model has basically the same complexity as solving the DSP. 

\subsection{Relation between LSP and DLSP}
\label{ref: Relation between DLSP and LSP}
For the comparison of model predictions and labeling experiments with CFSE or BrdU, the LSP model has been introduced~\cite{BanksSut2010,LuzyaninaRoo2009,LuzyaninaRoo2007}. The state variable of the LSP denote the label density $M^{\LSP}(t,x)$ in the population. In general, the evolution of $M^{\LSP}(t,x)$ is modeled by the PDE
\begin{equation}
\begin{aligned}
	&\frac{\partial M^{\LSP}(t,x)}{\partial t} + \frac{\partial (\nu(x) M^{\LSP}(t,x))}{\partial x} =\\
	&\hspace{2cm}- \left( \alpha(t,x) + \beta(t,x) \right) M^{\LSP}(t,x) + 2 \gamma \alpha(t,x) M^{\LSP}(t,\gamma x),
\label{eq: LSP model}
\end{aligned}
\end{equation}
with initial condition $M^{\LSP}(0,x) \equiv M^{\LSP}_{\init}(x)$~\cite{BanksSut2010}. As this model allows for label dependent division and death rates, $\alpha(t,x)$ and $\beta(t,x)$, it is in this respect more general than the DLSP.

However, it is not obvious why the cell division or death rates should depend on the label concentration. If the experiments are performed at low label concentrations far from the toxic regime, the population dynamics should be independent of the labeling~\cite{LyonsPar1994,MateraLup2004}. In particular, complex dependencies of $\alpha(t,x)$ and $\beta(t,x)$ on the label concentrations $x$, like those shown in~\cite{BanksSut2010}, are hard to argue. Additionally, a recent study supports that the introduced nonlinearities are correlated with the division number~\cite{BanksSut2010}.

Therefore, we just consider division and death rates which solely depend on time $t$, $\alpha(t)$ and $\beta(t)$. As proven in Appendix~\ref{app sec: Solution of DLSP solves LSP}, for this case, the solution $M(t,x)$ of the DLSP, with $\alpha_i(t) = \alpha(t)$ and $\beta_i(t) = \beta(t)$ and $N_{0,\init}(x) \equiv M^{\LSP}_{\init}(x)$, is equivalent to $M^{\LSP}(t,x)$. This shows that under these assumptions, the information provided by the LSP is a subset of the information available from the DLSP. This renders the DLSP more useful, as also subpopulation sizes are accessible.

Furthermore, for time dependent $\alpha(t)$ and $\beta(t)$, the solution of the DLSP can be approximated by a low-dimensional ODE system (Theorem~\ref{theorem: convergence} and~\ref{theorem: truncation error}). Hence, instead of computing $M^{\LSP}(t,x)$ using a PDE solver as done in all available publications, one may solve only a low-dimensional ODE system. Using the analytical results for the ODE system~\eqref{eq: ODE part of ansatz} even analytical solutions are available, e.g.,
\begin{equation}
	M^{\LSP}(t,x) = e^{-(\alpha + \beta)t} e^{kt}  \left(\sum_{i \in \mathbb{N}_0} \frac{(2 \alpha \gamma t)^i}{i!} M^{\LSP}_{0}(0, \gamma^i e^{kt} x) \right),
\label{eq: solution to LSP model}
\end{equation}
for constant rates $\alpha$ and $\beta$. Although this result for the LSP may be helpful to study various systems, we have not found it in the literature yet. The reason might be that a direct derivation of~\eqref{eq: solution to LSP model} is rather complex, whereas the study of the DLSP renders it straightforward.

Clearly, label dependent cell division and death rates or constant label loss rates were not considered here, in contrast to what was done in~\cite{BanksSut2010,LuzyaninaRoo2009,LuzyaninaRoo2007}. This was avoided as the decomposition of the solution shown in Section~\ref{subsec: Solution of the DLSP via decomposition} becomes impossible and solving the DLSP model gets computationally challenging. Nevertheless, the loss of these degrees of freedom is compensated by allowing for biologically more plausible division dependent cell parameters in the DLSP.

\subsection{DLSP as a unifying modeling framework}
The implications of the findings in Section~\ref{ref: Relation between DLSP and EGM}-\ref{ref: Relation between DLSP and LSP} are that the three most prevalent classes of population models are captured by the DLSP. Furthermore, it is more general, as label distributions and division dependent parameters may be considered, which are both important and well motivated from a biological point of view. Figure~\ref{fig: relation between model} illustrates the relations and shows how the EGM, the DSP, and the LSP may be constructed from DLSP via marginalization.

In contrast to the generality, the simulation effort increases only marginally when studying the DLSP instead of the DSP or the LSP. This is due to the decomposition into a system of ODEs (which is equivalent to the DSP), and a single set of PDEs. The set of PDEs can be solved analytically, and in several cases even analytical solutions for the ODE exist, facilitating an analytical solution of the overall system. Such analytical solutions can then be used to determine previously unknown analytical solutions for DSP and LSP, e.g., like~\eqref{eq: solution to LSP model}.

\begin{figure}[t!]
\centering
	\includegraphics[scale=1]{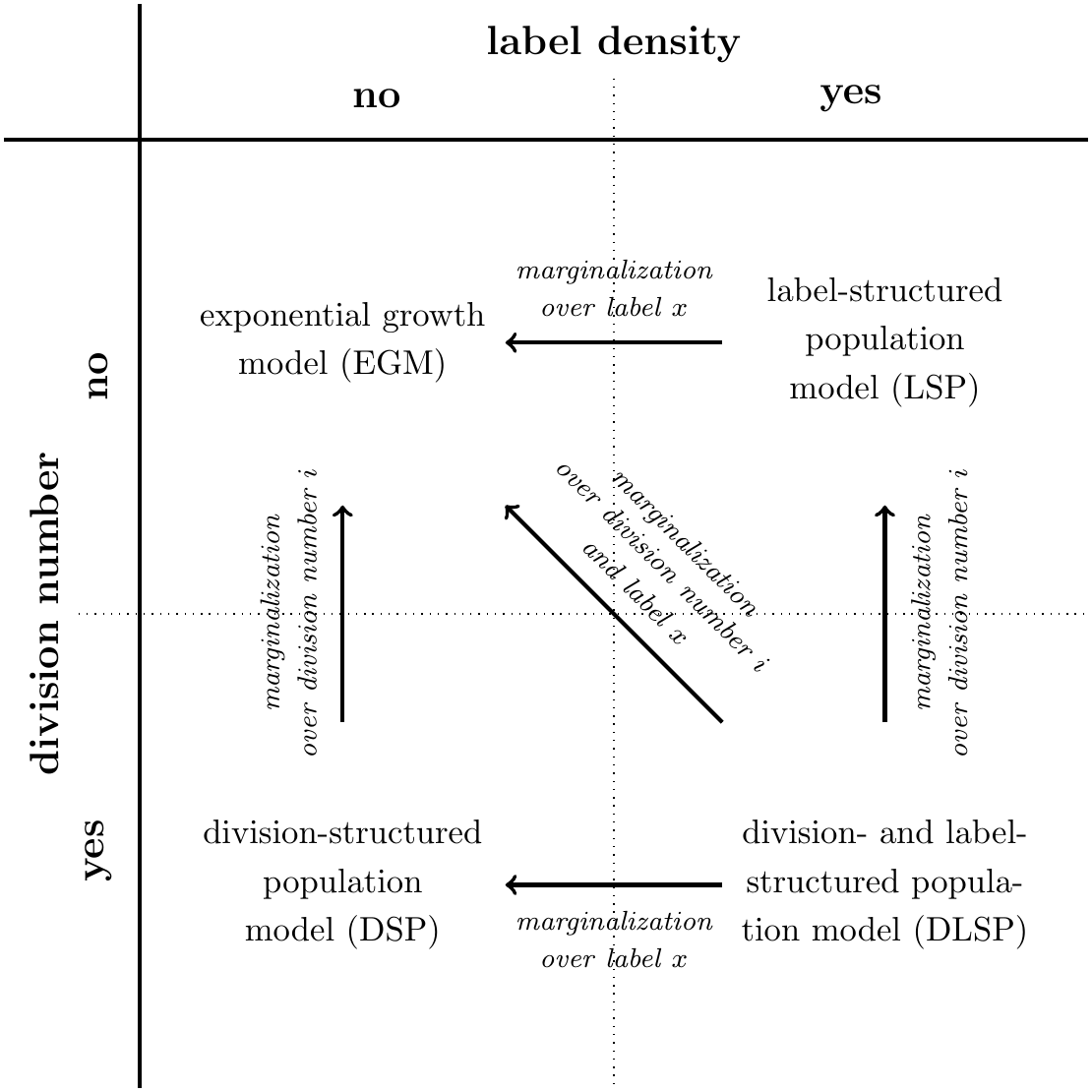}
\caption{Illustration of the relation between the exponential growth model (EGM), the division-structured population model (DSP), the label-structured population model (LSP), and the division- and label structured population model (DLSP). The models are distinguished using two properties, the availability of division numbers (vertical axis) and of information about the label distribution (horizontal axis). Arrows indicate whether and arrow labels describe how a model can be obtained from another model. It is apparent that the DLSP model is the most general model, as all remaining models can be constructed from it via marginalization.}
\label{fig: relation between model}
\end{figure}

\begin{remark}
Obviously, there exist extensions of the LSP and the DSP which are not captured by the current version of the DLSP. Examples are the aforementioned label concentration dependent division and death rates for the LSP~\cite{LuzyaninaRoo2007,LuzyaninaRoo2009,BanksSut2010} as well as DSP models with recruitment delay~\cite{LeonFar2004,DeBoerGan2006}. While the DLSP model can easily be extended to take such effects into account, the numerical analysis will get more challenging.
\end{remark}

\section{Computation of measured label distribution}
\label{sec: Computation of measured label distribution}

In the last section, we related the division- and label-structured population model to existing models. In this section, the prediction of the DLSP models will be related to data collected in proliferation assays.

\subsection{Autofluorescence and measured overall label distribution}
As outlined in the introduction, to obtain quantitative information about the proliferation dynamics, the fluorescent levels of individual cells are assessed using flow cytometry \cite{HawkinsHom2007}. The fluorescence level of an individual cell, $y \in \mathbb{R}_+$, summarizes the label induced fluorescence, $x$, and the autofluorescence, $x_a$,
\begin{equation}
y = x + x_a.
\end{equation}
The background, which might be interpreted as measurement noise, avoids a precise reconstruction of the label concentration. Furthermore, it limits the number of cell divisions which can be observed. While the label induced fluorescence, $x$, halves at cell division, this is not true for the autofluorescence. As the initial label concentration cannot be arbitrary high to avoid interference with the cell's functionality and toxicity, even for highly optimized labeling strategies only six to eight division can be observed before the observed fluorescence becomes indistinguishable from the background fluorescence~\cite{HawkinsHom2007}.

To address these problem a modified label-structured population model is introduced in \cite{BanksBoc2011} for the case of constant background fluorescence, $x_a$. This modified label-structured population model directly describes the evolution of $y$, accounting for the facts that (1) only $x$ is divided among daughter cells and (2) only $x$ is degraded over time. Unfortunately, this complicates the numerical treatment -- for this model no analytical expression for the label evolution is known -- and does not allow for the a separate analysis of the contributions. Furthermore, experiments showed that the background fluorescence varies among cells~\cite{HawkinsHom2007}. The autofluorescence, also called background fluorescence, is a stochastic variable $x_a \sim p(x_a)$, which is independent of the level of label concentration. The distribution of $x_a$, $p(x_a): \mathbb{R}_+ \rightarrow \mathbb{R}_+$, with $\int_{\mathbb{R}_+} p(x_a) dx_a = 1$, can be assessed using control experiments~\cite{HawkinsHom2007,BanksThom2012,Thompson2012}.

In this section, we propose an approach to predict the measured distribution of fluorescence, while explicitly distinguishing label dynamics and measurement process. The label dynamics are described by the DLSP model and the measured distribution of fluorescence is simply the convolution of the label induced fluorescence, $M(t,x)$, and the autofluorescence distribution, $p(x_a)$,
\begin{equation}
M^y(t,y) = \int_{\mathbb{R}_+} M(t,x) p(y-x) dx.
\label{eq: computation of measured fluorescence distribution}
\end{equation}
Hence, the measured fluorescence distribution $M^y(t,y)$, which is a number density function, can be obtained by simulating~\eqref{eq: solution of overall label density} and computing the convolution integral~\eqref{eq: computation of measured fluorescence distribution}. This is comparable to results described in \cite{BanksThom2012,Thompson2012}, where $x_a$ is partially contained in the model. However, the decomposition of the computation of $M^y(t,y)$ in dynamics and measurement is far more intuitive than a combined model as in~\cite{BanksBoc2011,BanksThom2012,Thompson2012} which combines the effects.

\subsection{Efficient approximation of measured overall label distribution}

It has been shown that the overall label distribution, $M(t,x)$, can be computed efficiently using the simulation of a low-dimensional ODE model and the analytical solution of a simple PDE. Unfortunately, this efficiency is corrupted by the need for solving the convolution integral~\eqref{eq: computation of measured fluorescence distribution}. A repeated evaluation, as required for parameter estimation (see, e.g., \cite{BanksBoc2011}), results in a large computational burden.

To reduce the computational complexity, we propose an approximation for $\hat{M}^y(t,y)$ of $M^y(t,y)$ which can be computed without integration. To allow for this approximation, we assume that the initial condition is a weighted sum of log-normal distributions,
\begin{equation}
N_{0,\init}(x) = \bar{N}_{0,\init}\sum_{j=1}^J f^j \log\mathcal{N}(x|\mu^j_{\init},(\sigma^j_\init)^2)
\label{eq: log-normal initial condition}
\end{equation}
with fraction parameters $f^j \in [0,1]$, with $\sum_{j=1}^J f^j = 1$, parameters $\mu^j_{\init},\sigma^j_\init \in \mathbb{R}_+$, and
\begin{equation}
\log\mathcal{N}(x|\mu,\sigma^2) = 
\left\{
\begin{array}{ll}
\frac{1}{\sqrt{2 \pi} \sigma x} e^{-\frac{1}{2}\left(\frac{\log(x)-\mu}{\sigma}\right)^2} &, x > 0 \\
0 &, x \leq 0.
\end{array}
\right.
\end{equation}
The faction parameters, $f^j$, determine which fraction of cells belongs to which log-normal distribution. The number of different log-normal distributions is denoted by $J \in \mathbb{N}$. In addition, we restrict the measurement noise to be log-normally distributed, $p(x_a) = \log\mathcal{N}(x_a|\mu_a,\sigma_a^2)$. These two assumptions are not restrictive, as any smooth distribution can be approximated arbitrarily well by a sum of log-normal distributions and as autofluorescence levels are known to be approximately log-normally distributed (see, e.g., \cite{HawkinsHom2007}).

Given~\eqref{eq: log-normal initial condition}, it can be shown that the label distribution in the individual subpopulation is
\begin{equation}
N_i(t,x) = \bar{N}_i(t) \sum_{j=1}^J f^j \log\mathcal{N}(x|\mu^j_{i}(t),(\sigma^j_\init)^2)
\end{equation}
with $\mu^j_{i}(t) = - i \log(\gamma) - \int_0^t k(\tau) d\tau + \mu^j_{\init}$  (for proof see Appendix~\ref{app sec: Proof that the PDE conserves log-normal distributions}). This follows directly from the analytical solution of $n_i(t,x)$. Thus, log-normal distributions are conserved under the considered class of partial differential equations, and log-normal initial conditions result in log-normal label distribution for $t > 0$. This implies that also the label induced fluorescence distribution is a sum of log-normal distributions,
\begin{align}
M(t,x)
= \sum_{i \in \mathbb{N}_0} N_i(t,x)
= \sum_{i \in \mathbb{N}_0} \bar{N}_i(t) \sum_{j=1}^J f^j \log\mathcal{N}(x|\mu^j_{i}(t),(\sigma^j_\init)^2)
\end{align}
By inserting this in the convolution integral~\eqref{eq: computation of measured fluorescence distribution}, we obtain by linearity of integration
\begin{align}
M^y(t,y)
&= \sum_{i \in \mathbb{N}_0} N_i^y(t,x) dx = \sum_{i \in \mathbb{N}_0} \int_{0}^\infty N_i(t,x) \log\mathcal{N}(y-x|\mu_a,\sigma_a^2) dx \\
&= \sum_{i \in \mathbb{N}_0} \bar{N}_i(t) \sum_{j=1}^J f^j \int_{0}^\infty \log\mathcal{N}(x|\mu^j_{i}(t),(\sigma^j_\init)^2) \log\mathcal{N}(y-x|\mu_a,\sigma_a^2) dx.
\end{align}
The individual summands of $M^y(t,y)$, $N_i^y(t,x) = \int_{0}^\infty N_i(t,x) \log\mathcal{N}(y-x|\mu_a,\sigma_a^2) dx$, are the measured fluorescence distributions in the subpopulations defined by a common division number. Therein, the summands of $N_i^y(t,x)$,
\begin{equation}
n_i^{y,j}(t,x) = \int_{0}^\infty \log\mathcal{N}(x|\mu^j_{i}(t),(\sigma^j_\init)^2) \log\mathcal{N}(y-x|\mu_a,\sigma_a^2) dx,
\end{equation}
describe the contribution of the $j$-th log-normal distribution in the initial condition to $n_i^y(t,x)$. This can be traced back as the superposition principle holds. Apparently, the efficient assessment of $M^y(t,y)$ is possible, using an efficient computational scheme for computing $n_i^{y,j}(t,x)$.

The probability density $n_i^{y,j}(t,x)$ is the probability density of the sum of two log-normally distributed random variables. Although, this density is of interest in many research fields (see \cite{Fenton1960,Beaulieu2004} and references therein), no analytical formula for computing $n_i^{y,j}(t,x)$ is known. Still, several approximations are available. One of the most commonly used approximation has been proposed by Fenton \cite{Fenton1960}. Fenton employs the fact that  although the distribution of the sum of two log-normally distributed random variables is not log-normal, it can still be closely approximated by a log-normal distribution. In \cite{Fenton1960}, this approximating log-normal distribution is chosen to have the same first two central moments, mean $\mathrm{E}^{j,y}_i(t)$ and variance $\mathrm{Var}^{j,y}_i(t)$, as the actual distribution of the sum.

The time-dependent central moments of $n_i^{y,j}(t,x)$ are the sums
\begin{align}
\mathrm{E}^{j,y}_i(t) &= \mathrm{E}^{j}_i(t) + \mathrm{E}_a, \\
\mathrm{Var}^{j,y}_i(t) &= \mathrm{Var}^{j}_i(t) + \mathrm{Var}_a,
\end{align}
of the time-dependent central moments of the label distribution of the $i$-th subpopulation, $\mathrm{E}^{j}_i(t)$ and $\mathrm{Var}^{j}_i(t)$, and the static autofluorescence, $\mathrm{E}_a$ and $\mathrm{Var}_a$, as it is known from basic statistics \cite{GrinsteadSne1997}. These central moments are 
\begin{align}
\mathrm{E}^{j}_i(t) &= e^{\mu^j_i(t)} e^{\frac{(\sigma^j_\init)^2}{2}},\\
\mathrm{Var}^{j}_i(t) &= e^{2 \mu^j_i(t) + (\sigma^j_\init)^2} \left(e^{(\sigma^j_\init)^2} -1 \right).
\end{align}
for the label distribution and 
\begin{align}
\mathrm{E}_a(t) &= e^{\mu_a} e^{\frac{(\sigma_a)^2}{2}}, \\
\mathrm{Var}_a(t) &= e^{2 \mu_a + (\sigma_a)^2} \left(e^{(\sigma_a)^2} -1 \right).
\end{align}
for the measurement noise. Following~\cite{Fenton1960}, the log-normal distribution exhibiting the same overall mean and variance has parameters
\begin{align}
\hat{\mu}^{j,y}_i(t) &= \log(\mathrm{E}^{j,y}_i(t)) - \frac{1}{2}\log\left(\frac{\mathrm{Var}^{j,y}_i(t)}{\mathrm{E}^{j,y}_i(t)} + 1\right),\\
\hat{\sigma}^{j,y}_i(t) &= \sqrt{\log\left(\frac{\mathrm{Var}^{j,y}_i(t)}{\mathrm{E}^{j,y}_i(t)} + 1\right)},
\end{align}
yielding the approximation
\begin{equation}
	\hat{n}_i^{y,j}(t,x) = \log\mathcal{N}(x|\hat{\mu}^{j,y}_i(t),(\hat{\sigma}^{j,y}_i(t))^2)
\end{equation}
of $n_i^{y,j}(t,x)$. Own studies revealed (not shown), that this approximation is for narrow distributions almost indistinguishable from the true distribution. In particular, if one of the distribution becomes narrow, the approximation can be made arbitrary good. This is helpful, as the precise parameterization of the initial condition might be a degree of freedom, which can be used to regulate the approximation quality.

Given the approximation of $n_i^{y,j}(t,x)$, the approximation 
\begin{equation}
\hat{M}^y(t,y) = \sum_{i \in \mathbb{N}_0} \bar{N}_i(t) \sum_{j=1}^J f^j \log\mathcal{N}(x|\hat{\mu}^{j,y}_i(t),(\hat{\sigma}^{j,y}_i(t))^2).
\label{eq: approximation of the convolution integral}
\end{equation}
of the measured fluorescence distribution can be computed. This approximation is the sum of log-normal distributions those parameters can be computed analytically. Therefore, it merely requires the evaluation of the log-normal distribution at different points, which can be made fairly efficient using lookup tables. The approximation~\eqref{eq: approximation of the convolution integral} can be determined orders of magnitude faster than the actual convolution integral~\eqref{eq: computation of measured fluorescence distribution} used, e.g., in \cite{BanksThom2012,Thompson2012}. Apparently, this approximation can also be combined with the truncation introduced in the last section.

Similar to the actual value, the approximation $\hat{M}^y(t,y)$ might be employed to perform parameter estimation. There, $\hat{M}^y(t,y)$ is compared directly \cite{LuzyaninaRoo2007} or indirectly \cite{LuzyaninaRoo2009,BanksSut2010,BanksBoc2011} with the measured flow cytometry data. This enables the inference of the model parameters, for instance, proliferation and death rate. 

\begin{remark}
Parameter estimation for the DLSP model is beyond the scope of this work. We focus on the development of modeling and simulation tools for structured cell population, which might in a second step be employed to infer parameters.
\end{remark}

\section{Example: Population with division number dependent parameters}
\label{sec: Example}

To demonstrate the properties of the DLSP model, an illustrative simulation study is performed. Therefore, a hypothetical cell population system with division number dependent proliferation rates $\alpha_i$ is considered. The existence of division number dependent proliferation dynamics is known for many cell systems~\cite{DeBoerGan2006,Hayflick1965,KassemAnk1997}, whereas the magnitude of the effect varies between them. This example shall illustrate the power of the DLSP and the proposed numerical procedure and therefore does not focus on a particular biological system.

The hypothetical cell population is assumed to have an initial proliferation rate of $\tilde{\alpha} = 0.02$ [1/hour], corresponding to an initial doubling time of 35 hours. This initial proliferation rate changes upon cell division. It is assumed that the proliferation rate decreases exponentially, $\forall i: \alpha_i = \tilde{\alpha} e^{-\Delta_\alpha i}$ [1/hours], with \mbox{$\Delta_\alpha = 0.23$ [-]}. This rate law is based on the findings in~\cite{KassemAnk1997} and results in a reduction of the proliferation rate by a factor of 2 when proceeding through 3 generations, thus $\alpha_{i+3} = \frac{\alpha_i}{2}$. The cell death rate is set to a constant value, $\forall i: \beta_i = \beta = 0.001$ [1/hours]. Concerning the labeling, a log-normal initial label density $N_{0,\init}(x)$ is assumed, as observed in many studies, e.g.,~\cite{LuzyaninaRoo2009,LuzyaninaRoo2007,BanksSut2010}. The label dilution factor and the degradation rate are set to $\gamma=2$ [-] and $k = 0.003$ [UI/hour], respectively, in which UI denotes the unit of label intensity. The autofluorescence is assumed to be log-normally distributed with $\mu_a = 2.5$ and $\sigma_a=0.3$. All parameter values are comparable to those available in the literature~\cite{LuzyaninaRoo2009,LuzyaninaRoo2007,BanksSut2010}.

\begin{figure}[t!]
\centering
	\includegraphics[scale=0.75]{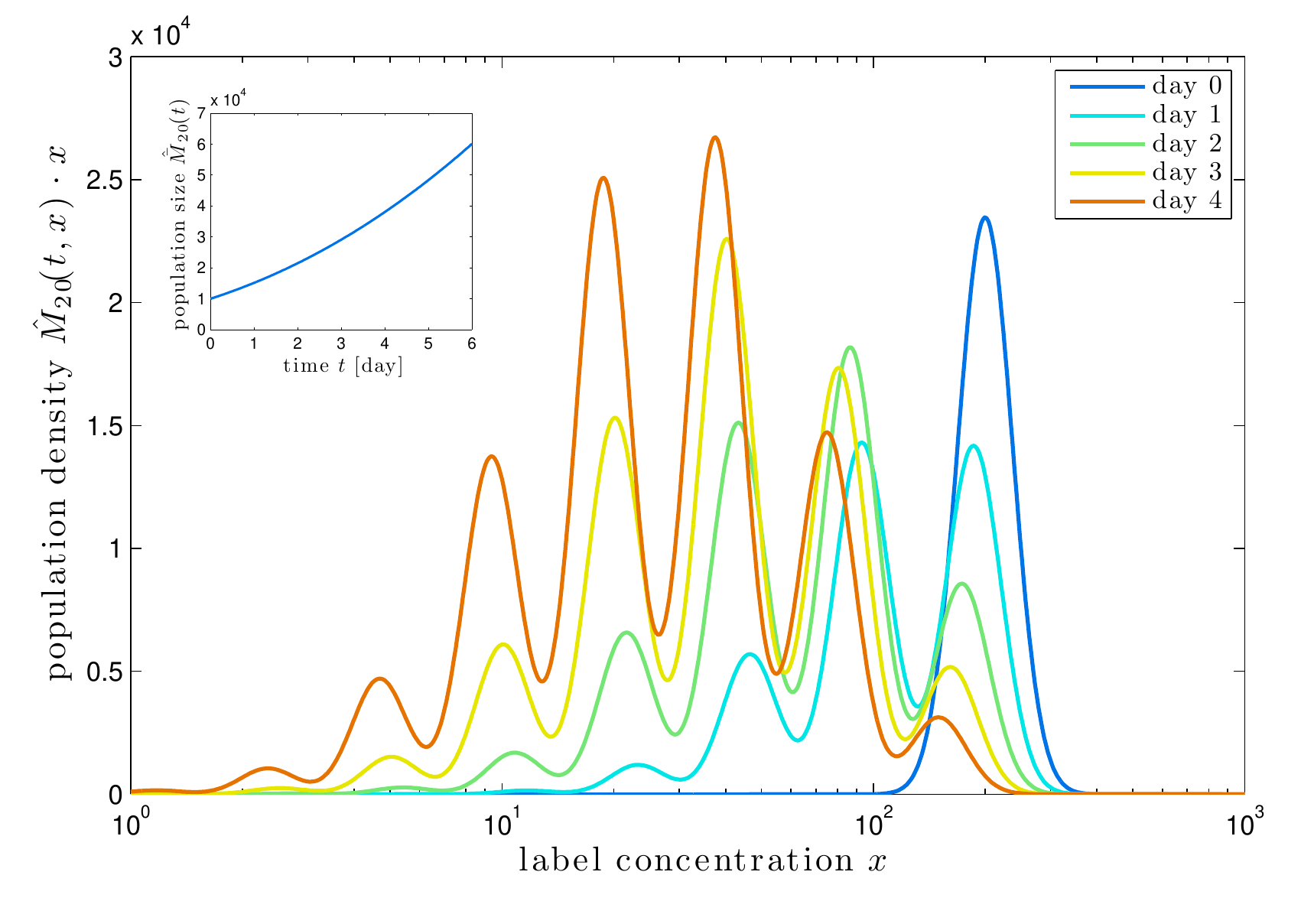}
\caption{Label density in cell populations at different points in time, computed from the first $S=20$ subpopulations. In order to ensure comparability with common histograms plots whose bins are logarithmically distributed, the label density is multiplied with the label concentration $x$.}
\label{fig-ex: overall population}
\end{figure}

The resulting cell population model is simulated for $t \in [0,6]$ days. The label density $\hat{M}_{20}(t,x)$ and the size $\hat{\bar{M}}_{20}(t)$ of the overall population are depicted in Figure~\ref{fig-ex: overall population}. Both quantities are computed using a truncation index of $S = 20$, thus merely the first 20 subpopulations are taken into account. This already ensures a small truncation error.

The actual truncation error and the bound for the truncation error are now studied in more detail. As no analytical solution is available, we compare all results to $\hat{M}_{100}(t,x)$. For this case, the analytical expression~\eqref{eq: upper bound for truncation error} for the error bound yields $E_S(t) \leq 10^{-20}$ over the whole time interval $t \in [0,6]$ days. Therefore, $\hat{M}_{100}(t,x)$ is considered as the exact solution. Given $\hat{M}_{100}(t,x)$ the truncation error is evaluated. From the results depicted in Figure~\ref{fig-ex: truncation error - exact} it is apparent that over the considered time interval, already $S=11$ provides an error smaller than $10^{-5}$. This illustrates that a small number of subpopulations is sufficient to obtain a good approximation of the population density. This result is also supported by the derived truncation error bound $E_S(T)$ (Figure~\ref{fig-ex: truncation error - bound}), while the expected truncation error is overestimated. This is visible for instance for $T = 6$ and $S = 11$, where the truncation error bound is several orders of magnitude higher than the actual truncation error.

\begin{figure}[t!]
\begin{center}
\subfigure[Truncation error, $\frac{||M(T,x) - \bar{M}_S(T,x)||_1}{||M(0,x)||_1}$.]{
	\includegraphics[scale=0.75]{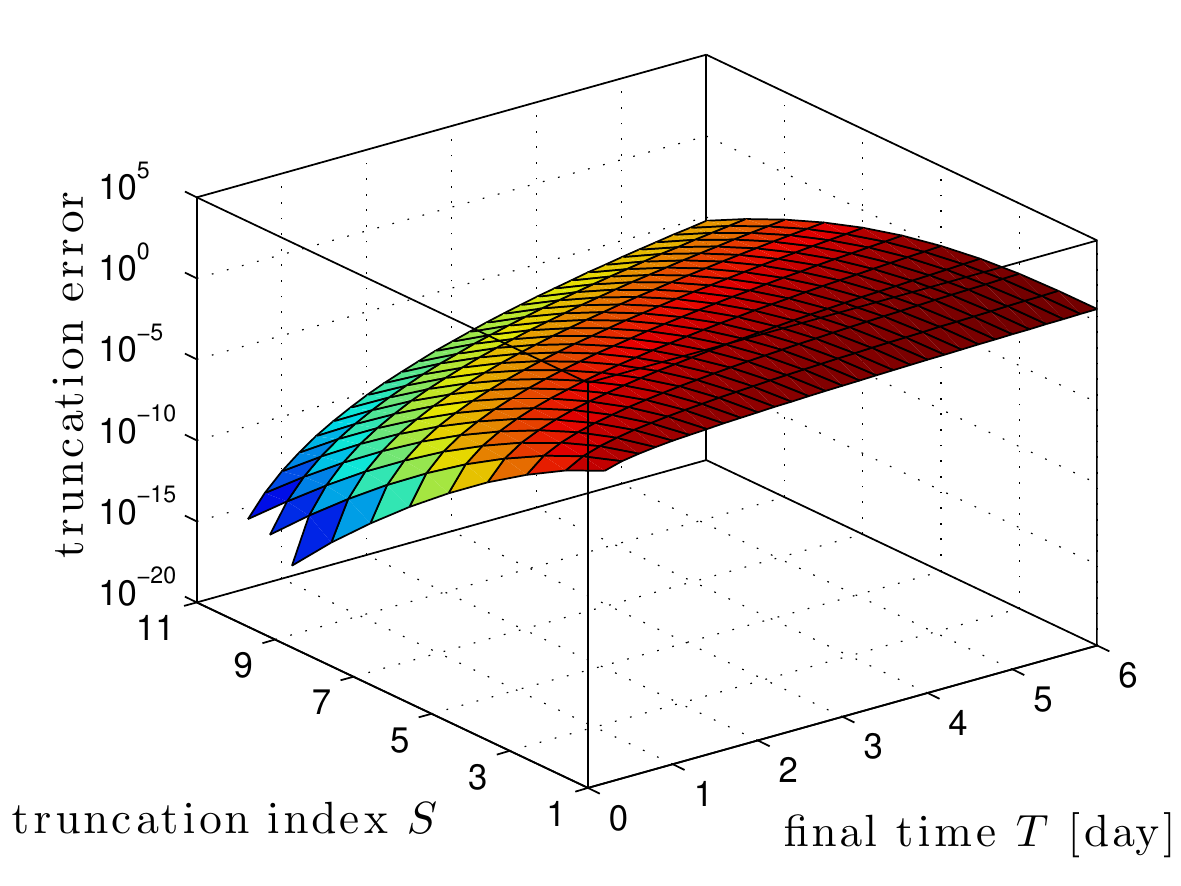}
\label{fig-ex: truncation error - exact}}
\\
\subfigure[Bound for the truncation error, $E_S(T)$.]{
	\includegraphics[scale=0.75]{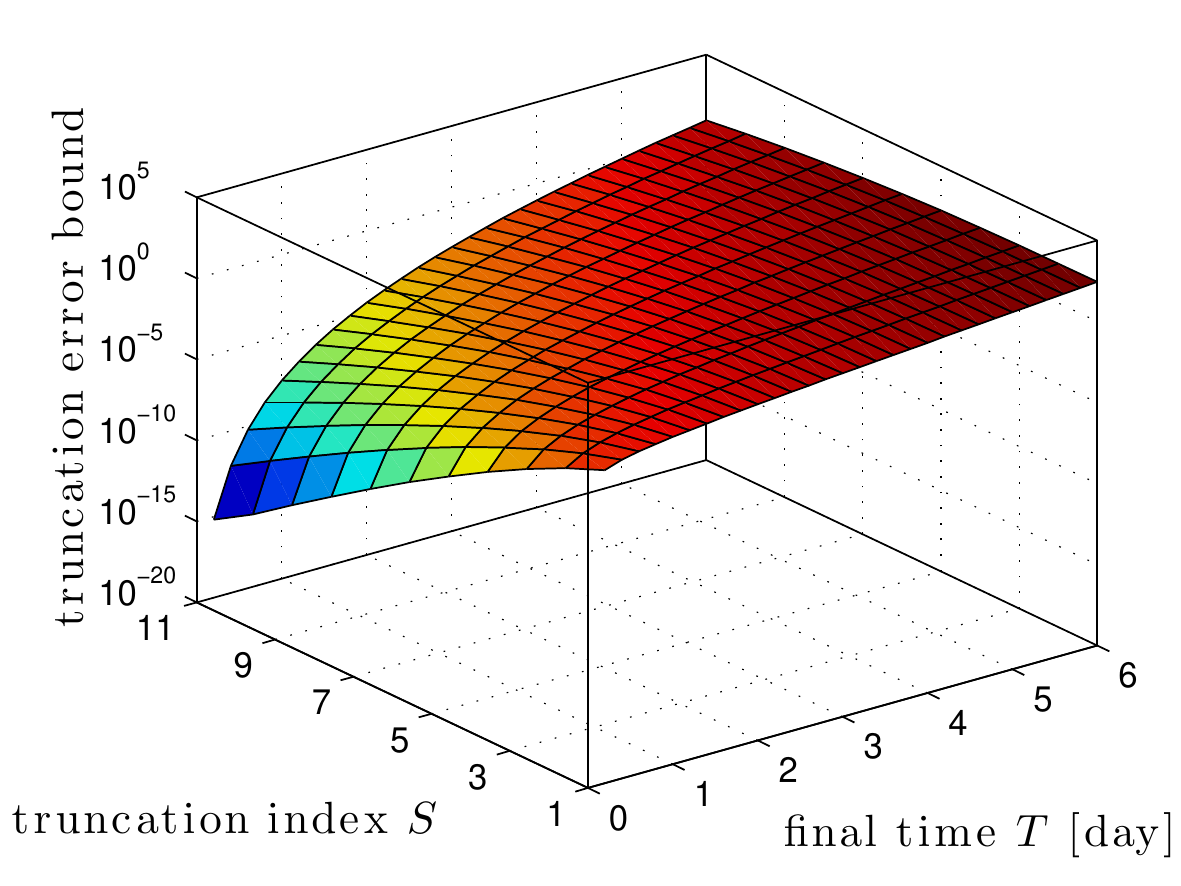}
\label{fig-ex: truncation error - bound}}
\caption{Truncation error~\subref{fig-ex: truncation error - exact} and truncation error bound~\subref{fig-ex: truncation error - bound} as function of the truncation index $S$ and the final time $T$.}
\label{fig-ex: truncation error 1}
\end{center}
\end{figure}

To assess the truncation error bound more precisely, we compute the minimal truncation index $S$ required to ensure a predefined error bound $\epsilon$. This analysis is performed using the exact truncation error (blue) and the truncation error bound (orange). The results are depicted in Figure~\ref{fig-ex: truncation error 2}, where Figure~\ref{fig-ex: truncation error - S time dependent} shows the time dependency and Figure~\ref{fig-ex: truncation error - S error dependent} shows the error level dependency of the minimal truncation index $S$. As verified previously, the computation from the exact truncation error yields lower truncation indices. The maximal observed difference for this system is a factor of 2. The difference increases over time and interestingly, for this system, the truncation index $S$ as a function of $T$ approximates a line with a slope of 2.5. While the slope is problem dependent, this effect has been observed for all considered systems. It probably originates from the structure of the truncation error bound~\eqref{eq: required truncation index}. Besides the time dependency, the index $S$ depends also on $\epsilon$. When $\epsilon$ is decreased by a factor of 10, the index $S$ has to increase by 2. This is a quite reasonable scaling and allows for very good approximations. For the system at hand, the analysis of the exact truncation error shows that $S=10$ ensures an error of 0.01~\% of the original population size. Employing the truncation error bound we compute $\bar{M}(t)$. Thus, the truncation error is overestimated by a factor of 2 but this number is computed without simulation and available even if the exact solution is not known. Given the upper bound of the truncation error, we can verify a priori that a very good approximation ($E_S(T) < 10^{-3}$) of the solution of the coupled system of PDEs~\eqref{eq: coupled PDE model for cell population} can be calculated by solving a system of 20 ODEs. This reduces the computational effort drastically.

\begin{figure}[!p]
\begin{center}
\subfigure[Truncation index $S$ required to ensure that $\frac{||M(T,x) - \bar{M}_S(T,x)||_1}{||M(0,x)||_1} \leq \epsilon$ (exact) and $E_S(T) \leq \epsilon$ (bound) for two levels of $\epsilon$.]{
	\includegraphics[scale=0.75]{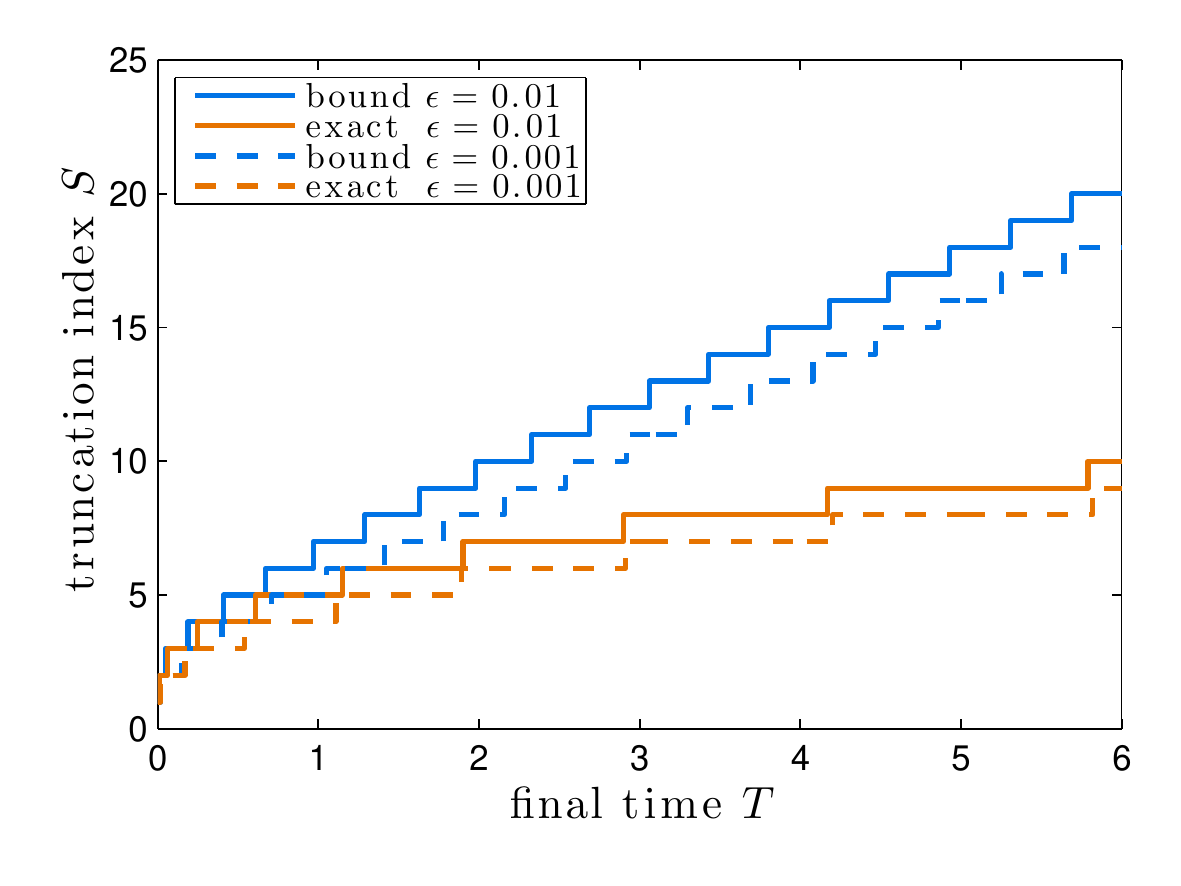}
\label{fig-ex: truncation error - S error dependent}}
\\
\subfigure[Truncation index $S$ required to ensure that $\frac{||M(T,x) - \bar{M}_S(T,x)||_1}{||M(0,x)||_1} \leq \epsilon$ (exact) and $E_S(T) \leq \epsilon$ (bound) for two points in time.]{
	\includegraphics[scale=0.75]{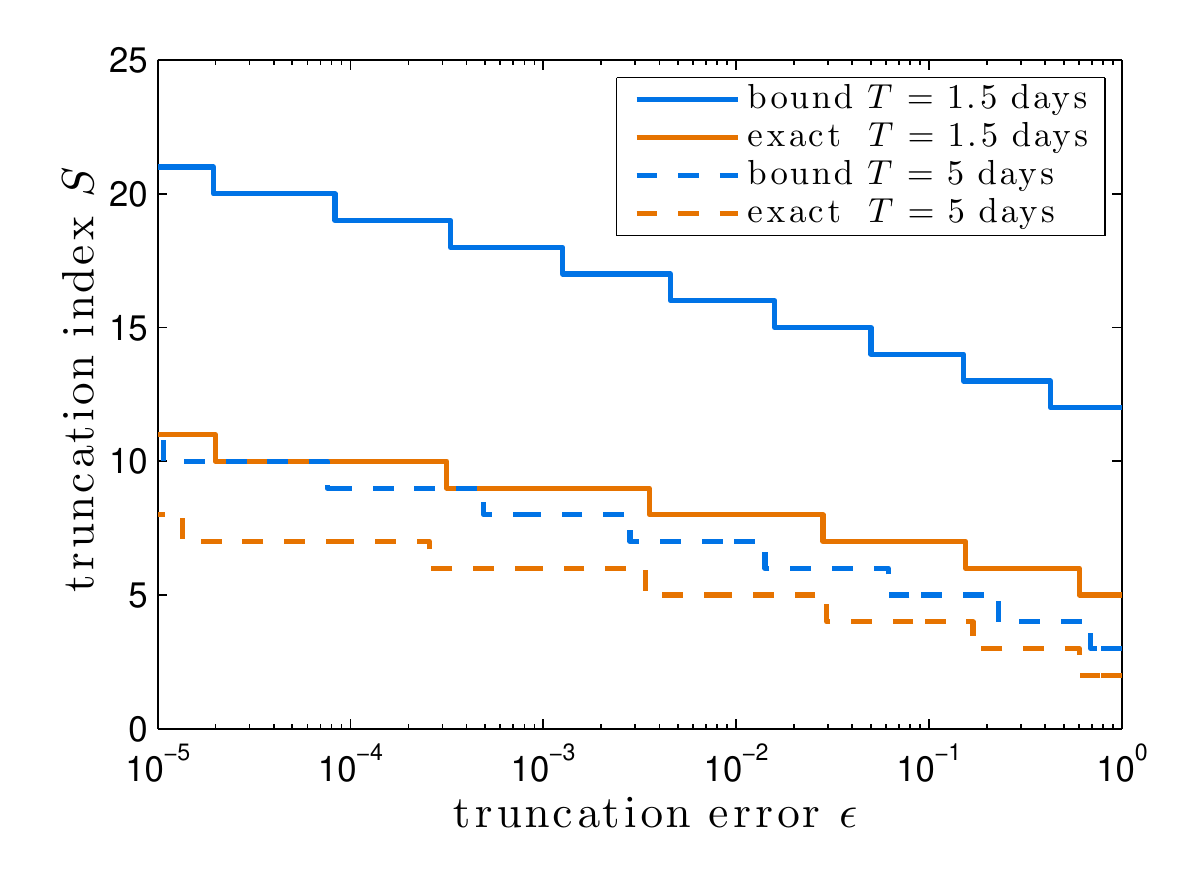}
\label{fig-ex: truncation error - S time dependent}}
\caption{Truncation index $S$ required to ensure a maximal error $\epsilon$ at a given time~$T$.}
\label{fig-ex: truncation error 2}
\end{center}
\end{figure}

Aside from the computational speed-up, the DLSP provides information about the overall label density and the size of the subpopulations. The former allows for the comparison of model prediction to labeling experiments, while the latter allows for the assessment of population properties like the mean division number (Figure~\ref{fig-ex: subpopulations}). These quantities are of interest in many studies, in which a precise understanding of the proliferation dynamics is of crucial importance.

\begin{figure}[!p]
\begin{center}
\subfigure[Size of subpopulations, $\bar{N}_i(t)$.]{
	\hspace*{2.2cm}\includegraphics[scale=0.75]{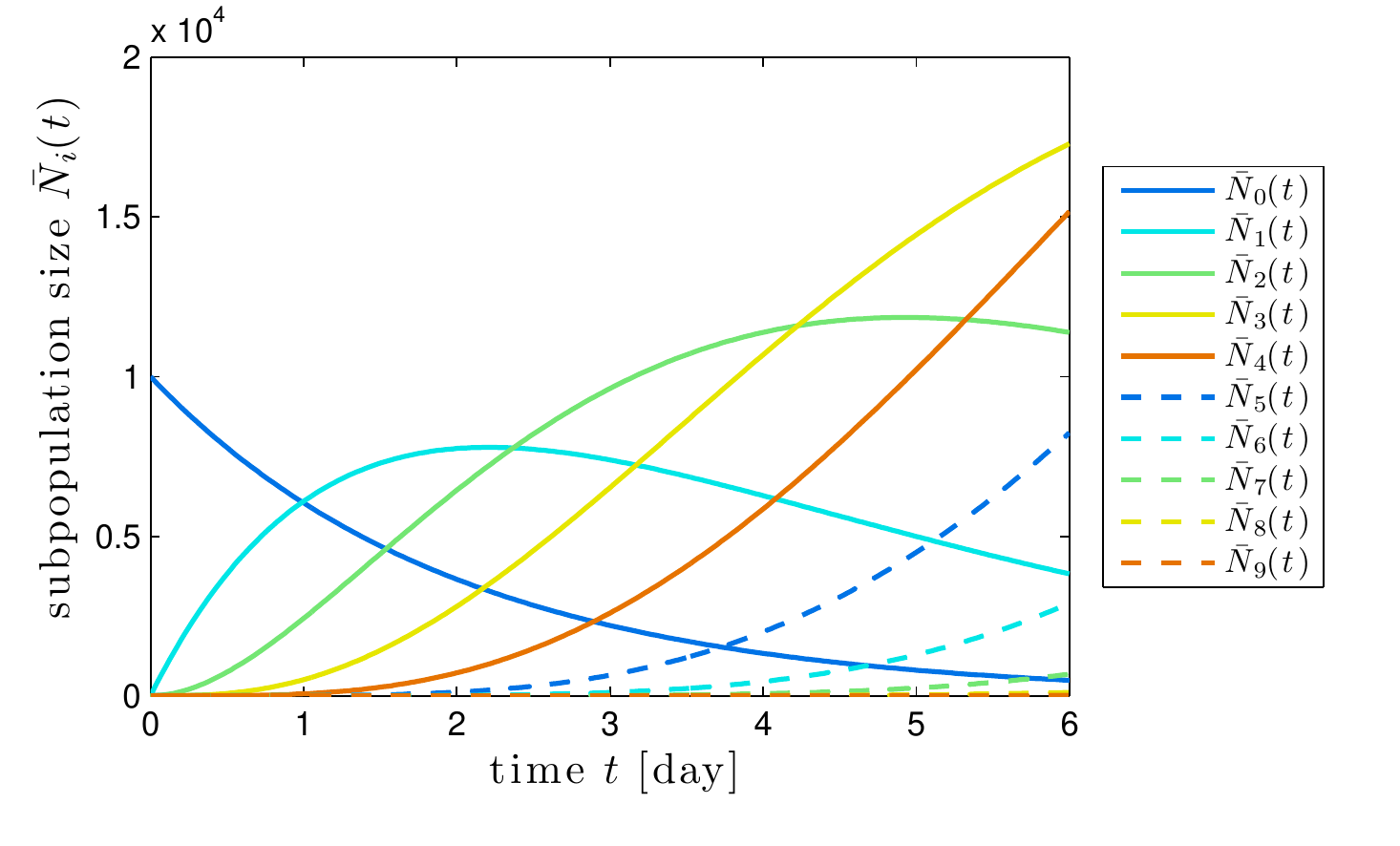}
\label{fig-ex: subpopulations - size}}
\subfigure[Mean number of divisions, $D(t) = \sum_{i\in\mathbb{N}_0} i \frac{\bar{N}_i(t)}{\bar{M}(t)}$.]{
	\includegraphics[scale=0.75]{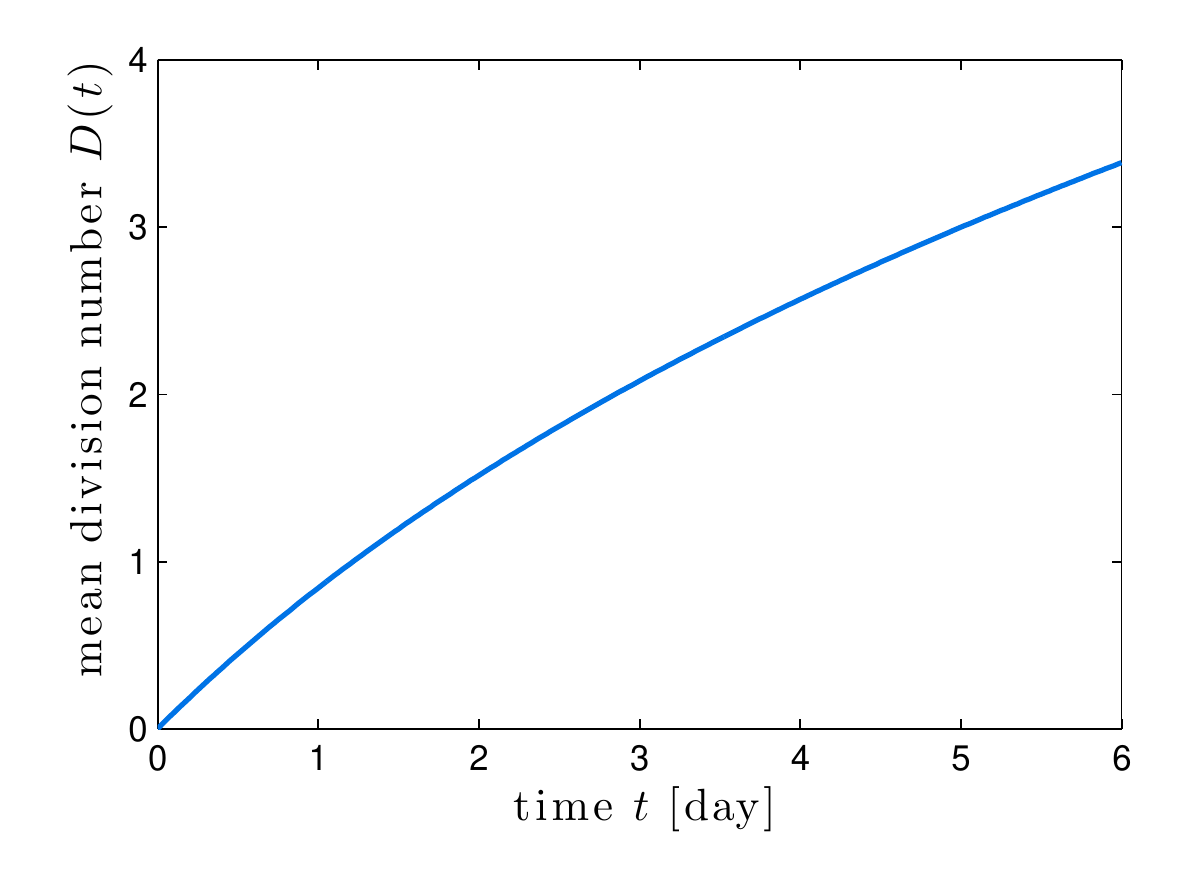}
\label{fig-ex: subpopulations - division number}}
\caption{Size of the subpopulations~\subref{fig-ex: subpopulations - size} and mean division number~\subref{fig-ex: subpopulations - division number} computed using the division- and label-structured population model.}
\label{fig-ex: subpopulations}
\end{center}
\end{figure}

Beyond the analysis of model properties, also an comparison of model prediction and measurement data is of interest. Therefore, the measured fluorescence distribution is required, which can be computed using~\eqref{eq: computation of measured fluorescence distribution} or approximated using~\eqref{eq: approximation of the convolution integral}. For the problem at hand, the true and the approximated solution are indistinguishable, while the approximated solution can be computed orders of magnitudes faster. The distribution of the measured fluorescence is depicted in Figure~\ref{fig-ex: measured overall population}. Similar to \cite{BanksBoc2011}, this simulation results shows that after a certain number of cell divisions, cells with different division numbers cannot be told apart any more. This is mainly caused by the halving of label concentration at each cell division, but also by the label degradation, resulting in an increased importance of the cellular autofluorescence.

\begin{figure}[t!]
\centering
	\includegraphics[scale=0.75]{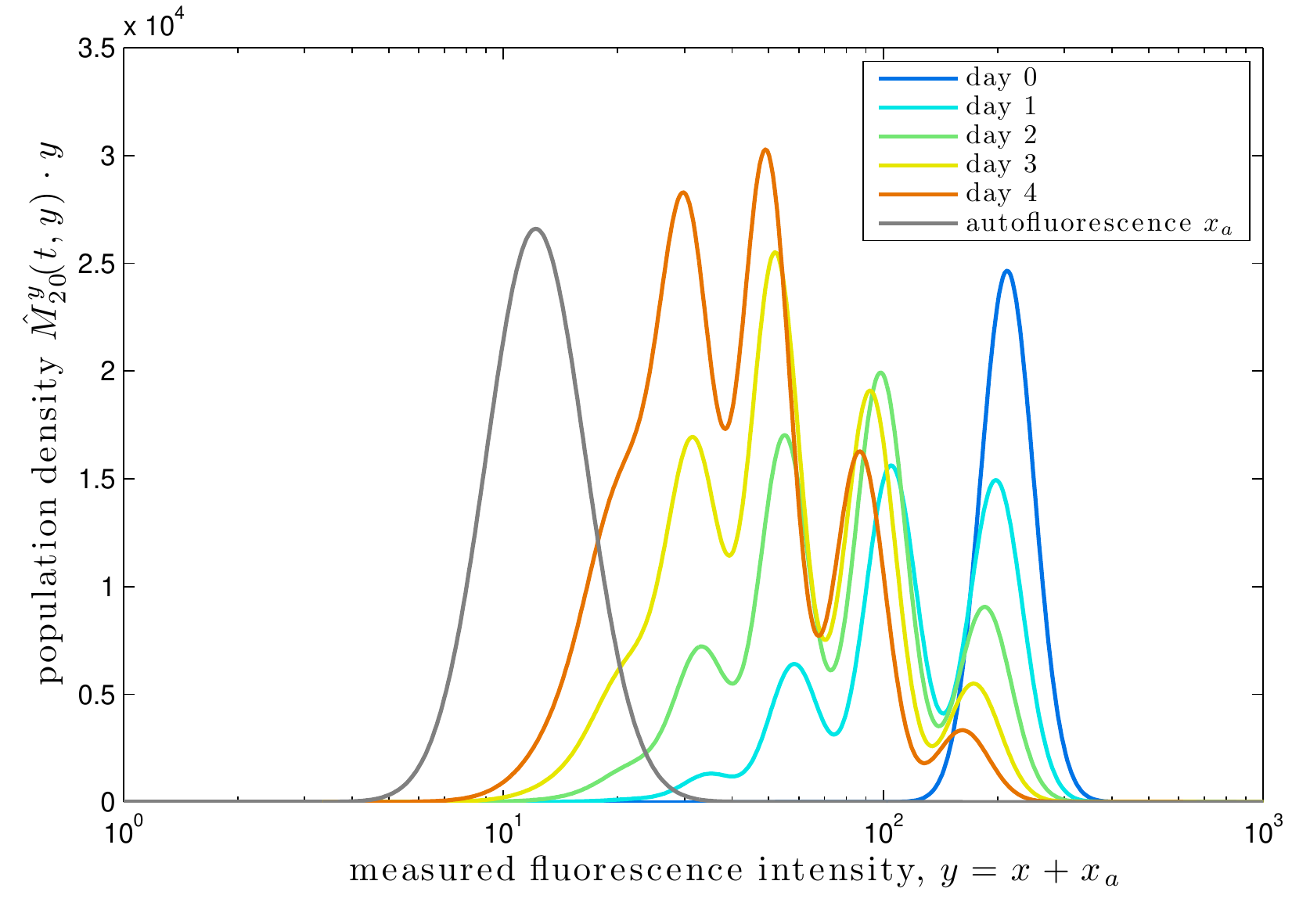}
\caption{Overall fluorescence intensity, $y = x + x_a$, in cell populations at different points in time, computed from the first $S=20$ subpopulations. In order to ensure comparability with common histograms plots whose bins are logarithmically distributed, the label density is multiplied with the label concentration $y$.}
\label{fig-ex: measured overall population}
\end{figure}

\section{Conclusion}
\label{sec: conclusion}

In this work, we have proposed a division- and label-structured population model which provides a unifying framework to study proliferating cell populations undergoing symmetric cell division. This model is based upon own work in \cite{SchittlerHas2011a} and considers both, continuous label dynamics and discrete division number dependent effects, such as cell aging. The resulting model is a system of coupled PDEs, which, under biologically plausible assumptions, can be split up into a system of ODEs and a set of decoupled PDEs. Each PDE describes the label distribution within one particular subpopulation and the ODE model describes the number of cells per subpopulation.

We have shown that the model is a generalization of existing division-structured population  models \cite{DeBoerGan2006,LuzyaninaMru2007} and label-structured population models \cite{LuzyaninaRoo2009,LuzyaninaRoo2007,BanksSut2010, BanksBoc2011}. Both model classes can be derived from the proposed model via marginalization. In contrast to these two existing types of models, the proposed model allows to incorporate division number dependent parameters as well as label distributions. The former one is important, as division number dependent parameters are found in many different cell systems and often are the subject of interest, while the latter one allows the direct comparison of model predictions and data. This supersedes complex and error-prone data analysis via deconvolution or peak detection~\cite{MateraLup2004,HawkinsHom2007,LuzyaninaMru2007}.

Clearly, though the model provides generalization and unification of several classes of population models, there remain models which are not covered. Examples are age-structured population models \cite{Marciniak-CzochraSti2009,StiehlMar2011,vonFoerster1959,Trucco1965,Oldfield1966,SinkoStr1967}, size-structured population models \cite{SinkoStr1967,DoumicMai2010}, and general population balance models \cite{Tsuchiya1966, FredricksonRam1967}. Furthermore, the size- and scar-structured population model for the asymmetrically dividing budding yeast has to be mentioned \cite{Gyllenberg1986}. There is quite a theory of population model construction introduced in \cite{DiekmannGyl1998}.

For the majority of these population models no analytical solutions are available. To study the dynamic properties of the models quantitatively, finite differences, finite volume, or finite elements discretization schemes are applied and the resulting ODE system is solved numerically (see, e.g., \cite{LuzyaninaRoo2007,BanksSut2010}). This need for numerical PDE solvers, which usually limits the state dimension to three to to the curse of dimensionality, is the main drawback of most populations models. It renders the analysis complex and partially accounts for the observed focus on steady state analysis \cite{Marciniak-CzochraSti2009,StiehlMar2011,DiekmannGyl2003,DiekmannGyl2010}, while dynamical aspects are mostly disregarded. Furthermore, an in-depth analysis of the model and its parameter has merely been performed for one-dimensional systems.

Besides its generality, the DLSP model can also be simulated efficiently. We have proven that the solution can be approximated by a low-dimensional ODE system, employing truncation. For the truncation error we have derived an a priori bound, which can be evaluated analytically. This lower bound can serve to determine the minimal model order/complexity required to achieve the desired approximation quality. This renders our model better applicable in cases where many other models, e.g.,~\cite{SinkoStr1967,BanksBoc2011}, come at a high computational cost. Also, these results can be used to allow a more rigorous reexamination of studies which employ the DLSP model, i.e., \cite{BanksThom2012,Thompson2012}.

In order to study the computational complexity, we have analyzed a cell population model with division number dependent parameters. Our study indicates that, if only the first eight divisions are of interest, which is the case in many studies~\cite{HawkinsHom2007}, the system of coupled PDEs can be approximated well by a system of 20 ODEs. The associated low computational complexity for evaluation of the model predictions facilitates the in-depth analysis of the population model. In particular, parameter estimation and uncertainty analysis becomes more efficient. Therefore, besides the novel biological insight which can be gained using the DLSP model, the developed decomposition and truncation scheme should be seen as a tool for future advanced estimation procedures. This is also the case for the proposed approach to determine the measured fluorescence distribution from the label distribution. The common convolution integral formulation could be employed, but the approximation employing the log-normal distribution is more efficient and yields almost identical results. This renders the proposed approximation a useful tool, enabling a more detailed study of the system. It has been shown in~\cite{HasenauerWal2011,HasenauerWal2011b,HasenauerLoh2012} that reformulations of the model and the objective function may allow for a significant speedup of the optimization. 

In subsequent studies, estimation methods and inverse problem formulations developed for exponential growth models~\cite{ZwieteringJon1990}, division-structured population models~\cite{DeBoerGan2006,LuzyaninaMru2007}, and label-structured population models~\cite{LuzyaninaRoo2009,LuzyaninaRoo2007,BanksSut2010,BanksBoc2011}, have to be adopted to apply to the DLSP model. This is also true for methods developed for size-structured populations \cite{DoumicMai2010,DoumicPer2009}, age-structured populations \cite{GyllenbergOsi2002} and general PDEs \cite{Banks1989}, for which even convergence properties have been established. Employing parameter estimation, e.g., for the T lymphocyte data published in~\cite{LuzyaninaRoo2009}, novel insights regarding division number dependencies on the population dynamics can be gained as indicated by \cite{BanksThom2012,Thompson2012} and shown by own unpublished results. In addition, due to the improved biological interpretation of the model, these results are expected to be far more reliable.

\section*{Acknowledgement}
The authors would like to acknowledge financial support from the German Research Foundation (DFG) within the Cluster of Excellence in Simulation Technology (EXC 310/1) at the University of Stuttgart, and from the German Federal Ministry of Education and Research (BMBF) within the SysTec program (grant nr. 0315-506A) and the FORSYS-Partner program (grant nr. 0315-280A).
D.S. acknowledges financial support by the MathWorks Foundation of Science and Engineering.








\newpage
\section*{Appendix}

\begin{appendix}

\section{Proof of analytical solution of PDE~\eqref{eq: PDE part of ansatz}}
\label{app sec: Analytical solution of PDE}

To determine the solution of the PDE~\eqref{eq: PDE part of ansatz} the method of characteristics \cite{Evans1998} is employed, which is possible as~\eqref{eq: PDE part of ansatz} is linear. The characteristics of~\eqref{eq: PDE part of ansatz} are defined by the ODEs
\begin{align}
\frac{d x}{d \tau} = -k(t) x, \quad 
\frac{d t}{d \tau} = 1, \quad
\frac{d n_i}{d \tau} &= k(t) n_i,
\label{eq: ODE for characteristics}
\end{align}
with $x(0) = x_0$, $t(0) = 0$, and $n_i(x_0) = n_{i,\init}(x_0)$. This system of ODEs has the solution
\begin{align}
x(\tau) = x_0 e^{-\int_0^\tau k(\tilde{\tau}) d\tilde{\tau}}, \quad
t(\tau) = \tau, \quad
n_i(\tau) = e^{\int_0^\tau k(\tilde{\tau}) d\tilde{\tau}} n_{i,\init}(x_0).
\end{align}
By substitution we obtain
\begin{equation}
\begin{aligned}
n_i(t,x) &= e^{-\int_0^t k(\tau) d\tau} n_{i,\init}(e^{-\int_0^t k(\tau) d\tau} x) \\
&= \gamma^i e^{-\int_0^t k(\tau) d\tau} n_{0,\init}(\gamma^i e^{\int_0^t k(\tau) d \tau}x) 
\end{aligned}
\end{equation}
as solution for~\eqref{eq: PDE part of ansatz}. \hspace*{\fill} $\square$

\section{Proof of Lemma~\ref{lem: ODE solution case 1}: Solution of ODE system}
\label{app sec: Analytical solution of ODE - Case 1}

In this section we prove by mathematical induction that the ODE system
\begin{equation}
\begin{split}
	i = 0: \hspace{2mm} &
	\frac{d \bar{N}_0}{d t} = - \left( \check{\alpha} + \beta \right) \bar{N}_0,\\
	\forall i \geq 1: \hspace{2mm} &
	\frac{d \bar{N}_i}{d t} = - \left( \check{\alpha} + \beta \right) \bar{N}_i + 2 \hat{\alpha} \bar{N}_{i-1}
\end{split}
\label{eq: bounding ode system for Lemma}
\end{equation}
with initial conditions $\bar{N}_0(0) = \bar{N}_{0,\init}$ and $\forall i \geq 1: \bar{N}_i(0) = 0$, has for $\hat{\alpha},\check{\alpha} \geq 0$ and $\beta > 0$ the solution:
\begin{align}
	\bar{N}_i(t) = \frac{(2 \hat{\alpha}t)^i}{i!} e^{-(\check{\alpha}+\beta)t} \bar{N}_{0,\init}.
	\label{eq: solution of bounding ode system}
\end{align}
Thereby,~\eqref{eq: bounding ode system for Lemma} is a generalization of~\eqref{eq:solution ODE case 1}.

It is trivial to verify that $\bar{N}_0$ and $\bar{N}_1$ are the solutions of~\eqref{eq: bounding ode system for Lemma} for $i=0$ and $i=1$, respectively. Hence, only the problem of proving that $\bar{N}_{k+1}$ is the solution of~\eqref{eq: bounding ODE part of ansatz} for $i=k+1$ given $\bar{N}_{k}$ remains. To show this, note that
\begin{align}
	\eqref{eq: solution of bounding ode system} \quad \laplace \quad
	\forall i \in \mathbb{N}_0: \bar{\mathcal{N}}_i = \frac{(2 \hat{\alpha})^i}{(s + \check{\alpha} + \beta)^{i+1}} \bar{N}_{0,\init},
\end{align}
in which $\bar{\mathcal{N}}_i$ is the Laplace transform of $\bar{N}_i$. Given this
\begin{equation}
\begin{aligned}
	\frac{d \bar{N}_{k+1}}{d t} &= - \left( \check{\alpha} + \beta \right) \bar{N}_{k+1} + 2 \hat{\alpha} \bar{N}_{k}\\
	\laplace \quad
	s \bar{\mathcal{N}}_{k+1} &= - \left( \check{\alpha} + \beta \right) \bar{\mathcal{N}}_{k+1} + 2 \hat{\alpha} \bar{\mathcal{N}}_{k} \\
	\Leftrightarrow \quad
	 \bar{\mathcal{N}}_{k+1} &= \frac{2 \hat{\alpha}}{s + \check{\alpha} + \beta} \bar{\mathcal{N}}_{k}.
\end{aligned}
\end{equation}
Substitution of $\bar{\mathcal{N}}_{k}$ now yields,
\begin{align}
	\bar{\mathcal{N}}_{k+1} &= \frac{(2 \hat{\alpha})^{k+1}}{(s + \check{\alpha} + \beta)^{k+2}} \bar{N}_{0,\init}
\end{align}
which by applying the inverse Laplace transformation concludes the mathematical induction and proves Lemma~\ref{lem: ODE solution case 1}. \hspace*{\fill} $\square$

\begin{remark}
Note that for $\check{\alpha} = \hat{\alpha} = \alpha$,~\eqref{eq: solution of bounding ode system} simplifies to~\eqref{eq:solution ODE case 1}. While for $\hat{\alpha} = \alpha_{\sup}$, $\check{\alpha} = \alpha_{\inf}$, $\beta = \beta_{\inf}$ , $\bar{N}_i = \bar{B}_i$, and $\bar{N}_{0,\init} = \bar{B}_{0,\init}$, we obtain the bounding system~\eqref{eq: bounding ODE part of ansatz} and its solution.
\end{remark}

\section{Proof of Lemma~\ref{lem: ODE solution case 2}: Solution of ODE system}
\label{app sec: Analytical solution of ODE - Case 2}

In this section we prove that if 
\begin{itemize}
\item $\forall i: \alpha_i(t) = \alpha_i \ \wedge \ \beta_i(t) = \beta_i$ and
\item $\forall i,j \in \mathbb{N}_0, i \neq j: \alpha_i + \beta_i \neq \alpha_j + \beta_j$
\end{itemize}
the solution of~\eqref{eq: ODE part of ansatz} is
\begin{align}
\begin{split}
	i = 0: \hspace{1mm} &
	\bar{N}_0(t) = e^{-(\alpha_0 + \beta_0)t} \bar{N}_{0,\init} \\
	\forall i \geq 1: \hspace{1mm} &
	\bar{N}_i(t) = 2^i \left(\prod_{j=1}^{i} \alpha_{j-1}\right) D_i(t) \bar{N}_{0,\init}
\end{split}
\label{eq: solution of ode system - case 2}
\end{align}
in which
\begin{align*}
	D_i(t) = \sum_{j=0}^{i}
	\left[
	\left(\prod_{\substack{k=0\\k\neq j}}^i ((\alpha_k + \beta_k) - (\alpha_j + \beta_j))\right)^{-1} e^{-(\alpha_j + \beta_j)t}
	\right].
\end{align*}
It is not difficult to verify that $\bar{N}_0$ and $\bar{N}_1$ are the solutions of~\eqref{eq: ODE part of ansatz} for $i=0$ and $i=1$, respectively. Hence, only the problem of proving that $\bar{N}_{k+1}$ is the solution of~\eqref{eq: bounding ODE part of ansatz} for $i=k+1$ given $\bar{N}_{k}$ remains. To show this, note that for 
\begin{align}
	\eqref{eq: solution of ode system - case 2} \quad \laplace \quad
	\forall i \in \mathbb{N}_0: \bar{\mathcal{N}}_i = 2^i \dfrac{\prod_{j=1}^{i} \alpha_{j-1}}{\prod_{j=0}^{i} (s + \alpha_j + \beta_j)} \bar{N}_{0,\init},
	\label{eq: laplace transform case 2}
\end{align}
in which $\bar{\mathcal{N}}_i$ is the Laplace transform of $\bar{N}_i$. The proof of this relation is provided in Appendix~\ref{app sec: Derivation of Laplace transform case 2}.

Given~\eqref{eq: laplace transform case 2} it follows that
\begin{equation}
\begin{aligned}
	\frac{d \bar{N}_{k+1}}{d t} &= - \left( \alpha_{k+1} + \beta_{k+1} \right) \bar{N}_{k+1} + 2 \alpha_{k} \bar{N}_{k}\\
	\laplace \quad
	s \bar{\mathcal{N}}_{k+1} &= - \left( \alpha_{k+1} + \beta_{k+1} \right) \bar{\mathcal{N}}_{k+1} + 2 \alpha_{k} \bar{\mathcal{N}}_{k} \\
	\Leftrightarrow \quad
	 \bar{\mathcal{N}}_{k+1} &= \frac{2 \alpha_{k}}{s + \alpha_{k+1} + \beta_{k+1}} \bar{\mathcal{N}}_{k}.
\end{aligned}
\end{equation}
Substitution of $\bar{\mathcal{N}}_{k}$ now yields,
\begin{align}
	\bar{\mathcal{N}}_{k+1} = 2^i \dfrac{\prod_{j=1}^{k+1} \alpha_{j-1}}{\prod_{j=0}^{k+1} (s + \alpha_j + \beta_j)} \bar{N}_{0,\init}
\end{align}
which by applying the inverse Laplace transformation concludes the mathematical induction and proves~\eqref{eq: solution of bounding ode system}. \hspace*{\fill} $\square$

\section{Derivation of Laplace transform~\eqref{eq: laplace transform case 2}}
\label{app sec: Derivation of Laplace transform case 2}

To derive~\eqref{eq: laplace transform case 2}, we study the partial fraction of
\begin{align}
\bar{\mathcal{N}}_i = 2^i \dfrac{\prod_{j=1}^{i} \alpha_{j-1}}{\prod_{j=0}^{i} (s + \alpha_j + \beta_j)} \bar{N}_{0,\init}.
\end{align}
As under the prerequisite $\forall i,j \in \mathbb{N}_0 \; \text{with} \; i \neq j: \alpha_i + \beta_i \neq \alpha_j + \beta_j$ all poles are distinct, the partial fraction can be written as
\begin{align}
\bar{\mathcal{N}}_i(s) &= 2^i \left(\prod_{j=1}^{i} \alpha_{j-1}\right) \left(\sum_{k=0}^N\frac{c_k}{(s+\alpha_k+ \beta_k)}\right) \bar{N}_{0,\init}.
\label{eq: partial fraction for case 2}
\end{align}
To determine the coefficients $c_k$, we consider the equality constraint
\begin{equation}
\begin{aligned}
\dfrac{1}{\prod_{j=0}^{i} (s + \alpha_j + \beta_j)} &= \sum_{k=0}^i\frac{c_k}{(s+\alpha_k+ \beta_k)} \\
\Leftrightarrow \hspace{2.7cm} 1 &= \sum_{k=0}^i c_k \prod_{\substack{j=1\\j\neq k}}^i (s + \alpha_j + \beta_j).
\end{aligned}
\end{equation}
As this equality constraint has to hold for all $s$, it must be satisfied for $s = - (\alpha_k+ \beta_k)$, yielding
\begin{equation}
	c_k = \left( \prod_{\substack{j=1\\j\neq k}}^i ((\alpha_j + \beta_j) - (\alpha_k + \beta_k)) \right)^{-1}.
\end{equation}
Given the values for $c_k$ one can easily verify~\eqref{eq: laplace transform case 2} by plugging in the $c_k$'s into~\eqref{eq: partial fraction for case 2}. Obviously, the proposed procedure can also be inverted, which concludes the derivation of~\eqref{eq: laplace transform case 2}. \hspace*{\fill} $\square$

\section{Proof of Theorem~\ref{theorem: convergence}: Convergence}
\label{app sec: Proof of convergence Theorem}

To prove Theorem~\ref{theorem: convergence}, the comparison theorem for series~\cite{Knopp1964} is applied. Therefore, we define the bounding system 
\begin{equation}
\begin{split}
	i = 0: \hspace{2mm} &
	\frac{d \bar{B}_0}{d t} = - \left( \alpha_{\inf} + \beta_{\inf} \right) \bar{B}_0,\\
	\forall i \geq 1: \hspace{2mm} &
	\frac{d \bar{B}_i}{d t} = - \left( \alpha_{\inf} + \beta_{\inf} \right) \bar{B}_i + 2 \alpha_{\sup} \bar{B}_{i-1}
\end{split}
\label{eq: bounding ODE part of ansatz} 
\end{equation}
with initial conditions
\begin{align*}
i = 0: \bar{B}_0(0) = \bar{N}_{0,\init}, \quad \forall i \geq 1: \bar{B}_i(0) = 0
\end{align*}
and $\alpha_{\inf}$, $\alpha_{\sup}$, and $\beta_{\inf}$ as in Theorem~\ref{theorem: convergence}. Due to the simple structure of~\eqref{eq: bounding ODE part of ansatz}, we can compute the analytical solution
\begin{align}
	\bar{B}_i(t) = \frac{(2 \alpha_{\sup}t)^i}{i!} e^{-(\alpha_{\inf}+\beta_{\inf})t} \bar{N}_{0,\init},
	\label{eq: solution of bounding ode system - main text}
\end{align}
whose derivation can be found in Appendix~\ref{app sec: Analytical solution of ODE - Case 1}.

The bounding system~\eqref{eq: bounding ODE part of ansatz} is obtained from~\eqref{eq: ODE part of ansatz} by reducing the outflows out of and increasing the inflows into the individual subpopulations. Intuitively, as the initial conditions of~\eqref{eq: bounding ODE part of ansatz} and~\eqref{eq: ODE part of ansatz} are identical and the right hand side of~\eqref{eq: bounding ODE part of ansatz} is for every $t \in [0,T]$ greater or equal than the right hand side of~\eqref{eq: ODE part of ansatz}, it follows that $B_i$ is an upper bound for $N_i$,
\begin{align}
	\forall t \in [0,T],i: \quad \bar{B}_i(t) \geq \bar{N}_i(t).
	\label{eq: bounding constraint}
\end{align}
This can be proven rigorously by applying M\"uller's theorem~\cite{Mueller1927}, as shown in~\cite{KiefferWal2011} for another system.

Given~\eqref{eq: solution of bounding ode system - main text} and~\eqref{eq: bounding constraint} one can prove the convergence of  $\sum_{i\in\mathbb{N}_0} N_i(t,x)$. To take into account that a distributed process is considered ($x\geq 0$), we study the maximum over $x$ and define
 	$B_i(t) := \bar{B}_i(t) \gamma^i e^{kt} n_x^{\sup} 
		= \frac{(2 \alpha_{\sup} \gamma)^i}{i!} t^i e^{-(\alpha_{\inf}+\beta_{\inf})t} e^{kt} N_x^{\sup}$
with $n_x^{\sup} := \sup_x \{ n_{0,\init}(x) \}$ and $N_x^{\sup} := \sup_x \{ N_{0,\init}(x) \}$. Thus, $B_i(t)$ is a point-wise upper bound of $N_i(t,x)$. For this definition of $B_i(t)$ it holds that
\begin{enumerate}
 \item[(i)] $\forall i,t,x \geq 0: \; 0 \leq N_i(t,x) \leq B_i(t) \; \forall i$, and
 \item[(ii)] the series 
 	\begin{equation}
 	\begin{aligned}
		\sum_{i=0}^{\infty} B_i(t)
		&= \left(\sum_{i=0}^{\infty} \frac{(2 \alpha_{\sup} \gamma t)^i}{i!}\right) e^{-(\alpha_{\inf}+\beta_{\inf})t} e^{kt} N_x^{\sup}
	\end{aligned} 
 	\end{equation}
	is convergent for every finite $t$.
\end{enumerate}
The latter one holds true as the series is simply the Taylor expansion of the exponential $e^{2 \alpha_{\sup} \gamma t}$. Under conditions (i) and (ii) it follows from the comparison theorem for series~\cite{Knopp1964} that the series $\sum_{i\in\mathbb{N}_0} \bar{N}_i(t)$ is convergent in $i$ for every $t \in [0,T]$ and for every $x \geq 0$. This concludes the proof. \hspace*{\fill} $\square$

\section{Proof of Theorem~\ref{theorem: truncation error}: Truncation error}
\label{app sec: Proof of truncation error Theorem}

To prove Theorem~\ref{theorem: truncation error}, note that
\begin{equation}
\begin{aligned}
||M(T,x) - \hat{M}_S(T,x)||_1
&= ||\sum_{i = S}^\infty \bar{N}_i(T) n_i(T,x)||_1\\
&= \sum_{i = S}^\infty \bar{N}_i(T) \int_{\mathbb{R}_+}n_i(T,x)dx\\
&= \sum_{i = S}^\infty \bar{N}_i(T),
\end{aligned}
\end{equation}
in which the individual lines follow from the approximation methods~\eqref{eq: approximated solution of overall label density}, the fact that all quantities are positive, and the definition of the normalized label intensity~\eqref{eq: PDE part of ansatz} which has unity integral for all times $T \geq 0$. The remaining term in the following is successively upper bounded, for which we employ the bounding system~\eqref{eq: bounding ODE part of ansatz}. As shown in Appendix~\ref{app sec: Proof of convergence Theorem}, it holds that $\bar{N}_i(t) \leq \bar{B}_i(t)$ which yields
\begin{equation}
\begin{aligned}
\sum_{i = S}^\infty \bar{N}_i(T) \leq \sum_{i = S}^\infty \bar{B}_i(T) = \sum_{i = S}^\infty \frac{(2 \alpha_{\sup}T)^i}{i!} e^{-(\alpha_{\inf}+\beta_{\inf})T} \bar{N}_{0,\init}.
\end{aligned}
\end{equation}
By completion of the sum, this can be written as
\begin{equation}
\begin{aligned}
\sum_{i = S}^\infty \bar{N}_i(T)
&\leq \left( e^{2 \alpha_{\sup} T} - \sum_{i = 0}^{S-1} \frac{(2 \alpha_{\sup}T)^i}{i!} \right) e^{-(\alpha_{\inf}+\beta_{\inf})T} \bar{N}_{0,\init}.
\end{aligned}
\end{equation}
Thus, by exploiting that $||M(0,x)||_1 = \bar{N}_{0,\init}$, one obtains~\eqref{eq: upper bound for truncation error}, which concludes the proof. \hspace*{\fill} $\square$

\section{Proof that the solution of LSP can be constructed from DLSP}
\label{app sec: Solution of DLSP solves LSP}

To prove that the DLSP provides the solution to the LSP, $M^{\LSP}(t,x) = M(t,x)$, we show that $M(t,x) = \sum_{i \in \mathbb{N}_0} \bar{N}_i(t) n_i(t,x)$ solves~\eqref{eq: LSP model}. Therefore, $M(t,x)$ is inserted in the left hand side $(*)$ of \eqref{eq: LSP model}, yielding
\begin{align*}
 	(*)
	& = \frac{\partial}{\partial t} \left(\sum_{i \in \mathbb{N}_0}{\bar{N}_i(t)n_i(t,x)}\right)
		- k \frac{\partial}{\partial x} \left(x \sum_{i \in \mathbb{N}_0}{\bar{N}_i(t)n_i(t,x)} \right) \\
	&= \sum_{i \in \mathbb{N}_0}{\left( \frac{d \bar{N}_i(t)}{d t} n_i(t,x)
		+ \bar{N}_i(t) \underbrace{\left( \frac{\partial n_i(t,x)}{\partial t}
		- k \frac{\partial(x n_i(t,x))}{\partial x} \right)}_{=0 \; \text{(with \eqref{eq: PDE part of ansatz})}} \right)}.
\end{align*}
In here, $d\bar{N}_i(t)/dt$ is substituted with \eqref{eq: ODE part of ansatz}, resulting in
\begin{align*}
 	(*)
	& = \sum_{i \in \mathbb{N}_0}{\left(-(\alpha(t) + \beta(t))\bar{N}_i(t) n_i(t,x) \right)}
		+ \sum_{i \in \mathbb{N}} 2\alpha(t)\bar{N}_{i-1}(t)
		\hspace{-3mm} \underbrace{n_i(t,x) } _{= \gamma n_{i-1}(t, \gamma x) }\\
	&= - (\alpha(t) + \beta(t)) \sum_{i \in \mathbb{N}_0}{\bar{N}_i(t) n_i(t,x)}
		+ 2 \gamma \alpha(t) \sum_{i \in \mathbb{N}_0}{\bar{N}_i(t) n_i(t, \gamma x)}
\end{align*}
This is equivalent to the result if $M(t,x)$ is inserted in the right hand side $(*)$ of~\eqref{eq: LSP model}. Hence, $M(t,x) = \sum_{i \in \mathbb{N}_0} \bar{N}_i(t) n_i(t,x)$ fulfills~\eqref{eq: LSP model} which concludes the proof. \hspace*{\fill} $\square$ 

\section{Proof that the PDE~\eqref{eq: PDE part of ansatz} conserves log-normal distributions}
\label{app sec: Proof that the PDE conserves log-normal distributions}

To prove that the PDE~\eqref{eq: PDE part of ansatz} conserves log-normal distributions, we use its analytical solution~\eqref{eq: solution of overall label density} and consider $n_{0,\init}(x) = \log\mathcal{N}(x|\mu_\init,\sigma_\init^2)$. This yields the solution
\begin{align}
n_i(t,x) &= \gamma^i e^{-\int_0^t k(\tau) d\tau} \log\mathcal{N}(\gamma^i e^{\int_0^t k(\tau) d \tau}x|\mu_\init,\sigma_\init^2).
\end{align}
Employing the definition of the log-normal distribution, this equation becomes
\begin{align}
n_i(t,x)
&= \gamma^i e^{-\int_0^t k(\tau) d\tau} \frac{1}{\sqrt{2 \pi} \sigma_\init \left(\gamma^i e^{\int_0^t k(\tau) d \tau}x\right)} e^{-\frac{1}{2} \left( \frac{\log\left( \gamma^i e^{\int_0^t k(\tau) d \tau}x \right) - \mu_\init}{\sigma_\init}\right)^2} \\
&= \frac{1}{\sqrt{2 \pi} \sigma_\init x} e^{-\frac{1}{2} \left( \frac{\log x - \left( - i \log \gamma - \int_0^t k(\tau) d \tau + \mu_\init\right)}{\sigma_\init}\right)^2}.
\end{align}
for  $x>0$, which can be restated as
\begin{align}
n_i(t,x) = \log\mathcal{N}(x|\mu_i(t),\sigma_\init^2),
\end{align}
in which $\mu_i(t) =  - i \log \gamma - \int_0^t k(\tau) d \tau + \mu_\init$. As this equation also holds for $x\leq0$, it follows that the log-normal distribution is conserved and merely the parameter $\mu$ is time dependent. Employing the superposition principle, this statement can be directly extended for sums of log-normal distributions, which concludes the proof.\hspace*{\fill} $\square$ 

\end{appendix}

\end{document}